%% file: lpa.tex
\documentclass{article}

\newcommand{\fstinst}{${}^1$}
\newcommand{\sndinst}{${}^2$}
\newcommand{\bothinst}{${}^{1,2}$}

\title{Homomorphisms and Minimality \\ for Enrich-by-Need \\
  Security Analysis}
\author{Daniel J. Dougherty\sndinst \and Joshua D. Guttman\bothinst
  \and John D. Ramsdell\fstinst \\[2mm]
  \fstinst The MITRE Corporation and \\
  \sndinst Worcester Polytechnic Institute}

\input{preamble}
\input{theorems}

\newif\ifpubrel

\ifpubrel
\usepackage{fancyhdr}
\fancypagestyle{title}{%
  \fancyhf{}%
  \fancyhead[L]{Approved for Public Release; Distribution Unlimited. Case Number 18-0186.}%
  \fancyfoot[C]{This technical data was produced for the
    U.S. Government\\ under Contract. No. W15P7T-13-C-A802, and is\\
    subject to the Rights in Technical Data-Noncommercial Items clause\\
    at DFARS 252.227-7013 (FEB 2012)\\
    \copyright 2018 The MITRE Corporation. ALL RIGHTS RESERVED.}%
}
\fi

\begin{document}
\maketitle
\ifpubrel
\thispagestyle{title}
\fi

\newpage
\tableofcontents
\newpage

\input{introduction}

\input{z3cpsa}

\input{solpa}
\input{doorsep}

\input{model-finding}

\input{results}
\input{concl}

\bibliographystyle{alpha}
\bibliography{../../inputs/secureprotocols}

\end{document}

%% file: preamble.tex
\usepackage{url}
\usepackage{amsmath}
\usepackage{amsfonts}
\usepackage{amssymb}
\usepackage{amsthm}

\usepackage{xspace}    
\usepackage{verbatim}  

\usepackage{graphicx}

\usepackage{hyperref} 
\hypersetup{%
  colorlinks=false,
}

\usepackage[all]{xy}

\newcount\timehh\newcount\timemm \timehh=\time
\divide\timehh by 60 \timemm=\time
\newcommand{\timestamp}{
  {\protect\small\sl\today\ --
    \ifnum\timehh<10 0\fi\number\timehh\,:\,
    \ifnum\timemm<10 0\fi\number\timemm}}

\newcommand{\pinch}{\vspace*{-5pt}}

\def\reflem#1{Lemma~\ref{#1}}

\def\refalg#1{Algorithm~\ref{#1}}

\newcommand{\eqdef}{\ensuremath{\stackrel{def}{=}}}

\newcommand{\set}[1]{\ensuremath{ \{ #1 \} }}
\newcommand{\dset}[2]{ \{ {#1} \mid {#2} \} }

\renewcommand{\phi}{\varphi}

\newcommand{\comp}{\circ}
\newcommand{\id}[1]{\mathop{\mathsf{id}}_{#1}\xspace}

\newcommand{\sig}{\ensuremath{\Sigma}\xspace}

\newcommand{\mdl}[1]{\ensuremath{\mathbb{{#1}}}\xspace}
\newcommand{\mM}{\mdl{M}}
\newcommand{\mN}{\mdl{N}}

\newcommand{\mK}{\mdl{K}}
\newcommand{\mA}{\mdl{A}}
\newcommand{\mB}{\mdl{B}}
\newcommand{\mC}{\mdl{C}}

\newcommand{\mF}{\mdl{F}}

\newcommand{\mP}{\mdl{P}}

\newcommand{\modclass}{\ensuremath{\mathcal{M}}\xspace}
\newcommand{\modcat}[1]{{\ensuremath{\mathcal{M}}_{#1}}\xspace}
\newcommand{\imodcat}[1]{\ensuremath{\mathcal{M}}^i_{#1}\xspace}
\newcommand{\modcatsig}{{{\ensuremath{\modcat}{\sig}}}\xspace}
\newcommand{\imodcatsig}{\ensuremath{\imodcat}{\sig}\xspace}

\newcommand{\ch}[1]{\ensuremath{\mathit{ch}_{#1}}}
\newcommand{\ich}[1]{\ensuremath{\mathit{ch}^{i}_{#1}}}
\newcommand{\diagram}[1]{\ensuremath{\Delta}_{#1}}
\newcommand{\posdiagram}[1]{\ensuremath{\Delta^{+}}_{#1}}

\newcommand{\flip}[1]{\ensuremath{\mathit{flip}_{#1}}}
\newcommand{\domain}[1]{\ensuremath{\delta}_{#1}\xspace}

\newcommand{\homleq}{\ensuremath{\precsim}\xspace}
\newcommand{\homle}{\ensuremath{\precnsim}\xspace}
\newcommand{\homeq}{\ensuremath{\thickapprox}\xspace}
\newcommand{\homgeq}{\ensuremath{\succsim}\xspace}
\newcommand{\ihomleq}{\ensuremath{\precsim^{i}}\xspace}
\newcommand{\ihomle}{\ensuremath{\precnsim^{i}}\xspace}
\newcommand{\ihomeq}{\ensuremath{\thickapprox^{i}}\xspace}
\newcommand{\ihomgeq}{\ensuremath{\succsim^{i}}\xspace}

\newcommand{\amin}{\ensuremath{a}-{minimal}\xspace}
\newcommand{\imin}{\ensuremath{i}-{minimal}\xspace}

\newcommand{\homTo}[1]{\ensuremath{\mathit{homTo}_{#1}}\xspace}
\newcommand{\homFrom}[1]{\ensuremath{\mathit{homFrom}_{#1}}\xspace}
\newcommand{\avoid}[1]{\ensuremath{\mathit{avoid}_{#1}}\xspace}
\newcommand{\rep}[1]{\ensuremath{\mathit{rep}_{#1}}\xspace}
\newcommand{\ehom}[1]{\ensuremath{\mathit{endo}_{#1}}\xspace}

\newcommand{\profile}{\ensuremath{\mathit{prf}}\xspace}
\newcommand{\prfbounds}{\ensuremath{\delta_{\profile}}\xspace}

\newcommand{\thy}[1]{\ensuremath{\mathcal{{#1}}}\xspace}
\newcommand{\thyT}{\thy{T}}
\newcommand{\thyG}{\thy{G}}
\newcommand{\Th}[1]{\ensuremath{\mathop{Th}(#1)}}

\newcommand{\cpsa}{\textsc{cpsa}}
\newcommand{\sas}{\textsc{sas}}
\newcommand{\lpa}{\textsc{lpa}}
\newcommand{\smt}{\textsc{smt}}
\newcommand{\smtlib}{\textsc{smt-lib}}
\newcommand{\sop}{\textsc{sop}}
\newcommand{\cn}[1]{\ensuremath{\operatorname{\mathsf{#1}}}}
\newcommand{\cnc}[1]{\ensuremath{\mathsf{#1}}}

\newcommand{\fn}[1]{\ensuremath{\operatorname{\mathit{#1}}}}
\newcommand{\thyr}[1]{\ensuremath{\operatorname{\mathrm{#1}}}}
\newcommand{\sdom}{\fn{Dom}}

\newcommand{\typ}{\mathbin:}
\newcommand{\typs}[1]{\typ\srt{#1}}
\newcommand{\seq}[1]{\ensuremath{\langle#1\rangle}}
\newcommand{\append}{\mathbin{\S}}
\newcommand{\enc}[2]{\ensuremath{\{\!|#1|\!\}_{#2}}}
\newcommand{\iv}[1]{\ensuremath{{#1}^{-1}}}
\newcommand{\inbnd}{\mathord -}
\newcommand{\outbnd}{\mathord +}
\newcommand{\srt}[1]{\ensuremath{\mathsf{#1}}}

\newcommand{\pos}{\ensuremath{\mathbb{Z}^+}}
\newcommand{\all}[1]{\mathop{\forall#1\mathpunct.}}
\newcommand{\some}[1]{\mathop{\exists#1\mathpunct.}}
\newcommand{\prefix}{\dagger}

\newcommand{\lang}{\Sigma}
\newcommand{\glang}{\ensuremath{\mathcal{L}}}
\newcommand{\mdlm}{\mathcal{M}}
\newcommand{\diag}{\mathcal{D}}
\newcommand{\kmdls}{\models_\Pi}
\newcommand{\alg}[1]{\ensuremath{\mathfrak#1}}
\newcommand{\alga}{\alg{A}}

\newcommand{\skel}{\mathsf{k}}

\newcommand{\truth}{\mathsf{true}}
\newcommand{\falsehood}{\mathsf{false}}
\newcommand{\nodes}{\ensuremath{\mathcal{N}}}
\newcommand{\evt}{\fn{evt}}
\newcommand{\msg}{\fn{msg}}

\newcommand{\iprec}{\ll}

\newcommand{\bundle}{\ensuremath{\mathcal{B}}}
\newcommand{\limp}{\ensuremath{\Longrightarrow}}
\newcommand{\qdot}{\,\mathbf{.}\;}

%% file: theorems.tex
\usepackage{amsthm}

\theoremstyle{plain} 
\newtheorem{theorem}{Theorem}
\newtheorem{lemma}[theorem]{Lemma}

\newtheorem{corollary}[theorem]{Corollary}

\theoremstyle{definition} 
\newtheorem{definition}[theorem]{Definition}
\newtheorem{notation}[theorem]{Notation}
\newtheorem*{notation-un}{Notation}

\newtheorem*{remark-un}{Remark}
\newtheorem{algo}[theorem]{Algorithm}
\newcounter{stepno}
\newenvironment{listalg}{
  \begin{list}{\arabic{stepno}.}{
      \usecounter{stepno}\itemsep=-0.5ex\topsep=0.5ex}}{\end{list}}
\theoremstyle{remark} 
\newtheorem{example}[theorem]{Example}
\newtheorem*{example-un}{Example}

\newtheorem*{examples-un}{Examples}

%% file: introduction.tex
\section{Introduction}


Cryptographic protocol analysis is a well-developed subject, with many
tools and rigorous techniques that can be used to determine what
confidentiality, authentication~(among
others,~\cite{cpsa16,DBLP:conf/fosad/EscobarMM07,Blanchet02,cremers2012operational}),
and indistinguishability
properties~(e.g.~\cite{Blanchet04,BlanchetAF08,ChadhaCCK16}) that a
protocol satisfies.

However, protocols are used to support security functionality required
by the applications that use those protocols.  Those applications may
have access control and authorization goals, and the real criterion
for whether a protocol is good enough for that usage is whether these
application-level requirements are always met.  For instance, in the
TLS resumption attacks~\cite{rfc5746},
cf.~\cite{BhargavanDFPS14,RoweEtAl2016}, the protocol did not allow
the server application to distinguish unauthenticated input at the
beginning of a data stream from subsequent authenticated input.  This
may lead to erroneous authorization decisions.

Another area in which application behavior affects protocol goals
concerns environmental assumptions.  For instance, some protocols fail
if the same long-term key is ever used by a principal when playing the
server role, and it has been used when playing a client role.
However, an application may ensure that no server ever needs to
execute the protocol in the client role at all.  This policy would
ensure that an otherwise weak protocol reliably supports the needs for
the application.

Logical Protocol Analysis is our term for combining a protocol
analyzer with these additional concerns, which we analyze via model
finding.  Our goal is to analyze cryptographic protocols that include
trust axioms that cannot be stated using the typical input to a
protocol analyzer such as {\cpsa}.

\paragraph{An Example:  DoorSEP.}  We turn next to a motivating
scenario, complete details of which may be found in
Section~\ref{sec:doorsep}.  We begin here by describing the protocol,
called the Door Simple Example Protocol (DoorSEP).  It is derived from
an expository protocol due to Bruno Blanchet~\cite{Blanchet08}, who
designed it to have a weakness.  We will illustrate conditions under
which, despite this weakness, the protocol achieves the needs of the
application.

Imagine a door~$D$ which is equipped with a badge reader, and a
person~$P$ who is equipped with a badge.  When the person swipes the
badge, the cryptographic protocol is executed.  The door and person
are each identified by the public part of an asymmetric key pair, with
\iv{D} and \iv{P} being the respective principal's private keys.  We
write \enc{M}{K} for the encryption of message~$M$ with key~$K$.  A
message~$M$ is signed by~$P$ by encrypting with $P$'s private key
($\enc{M}{\iv{P}}$).

The person initiates the exchange by creating a fresh symmetric
key~$K$, signing it, and sending it to the door encrypted with the
door's public key.  The door extracts the symmetric key after checking
the signature, freshly generates a token~$T$, and sends it to the
person encrypted with the symmetric key.  The person demonstrates they
are authorized to enter by decrypting the token and sending it as
plain text to the door.  DoorSEP may be expressed in Alice and Bob
notation:

\[
\begin{array}{l@{{}:{}}l}
  P\to D & \enc{\enc{K}{\iv{P}}}{D}\\
  D\to P & \enc{T}{K}\\
  P\to D & T.
\end{array}
\]

An analysis of DoorSEP by {\cpsa} shows an undesirable execution of
this protocol.  Assume the person's private key \iv{P} is
uncompromised and the door has received the token it sent out.  In
this situation, {\cpsa} will deduce that person~$P$ freshly created
the symmetric key~$K$.  However, there is nothing in this protocol to
ensure that the person meant to open door $D$.  If adversary~$A$
gets~$P$ to use compromised door~$D'$, the adversary can perform a
man-in-the-middle attack:

\[
\begin{array}{l@{{}:{}}l}
  P\to A & \enc{\enc{K}{\iv{P}}}{D'}\\
  A\to D & \enc{\enc{K}{\iv{P}}}{D}\\
  D\to A & \enc{T}{K}\\
  A\to D & T.
\end{array}
\]
Without additional assumptions, the door cannot authenticate the
person requesting entry.

But think about this situation:  Can we trust the person to swipe her
badge only in front of the door that matches the badge?  Can we ensure
that that door has an uncompromised private key?  If so, then the
adversary cannot exercise the flaw.

We regard this as a \emph{trust assumption}, and we can express it as
an axiom:
\begin{quote}
  If an uncompromised signing key ${\iv{P}}$ is used to prepare an
  instance of the first DoorSEP message, then its owning principal has
  ensured that the selected door $D$ has an uncompromised private key.
\end{quote}
The responsibility for ensuring the truth of this axiom may be split
between the person and the organization controlling the door.  The
person makes sure to swipe her badge only at legitimate doors of the
organziation's buildings.  The organization maintains a security
posture that protects the corresponding private keys.

Is DoorSEP good enough given the trust axiom?

To analyze this protocol with this trust assumption we rely on model
finding.  We provide a theory to Razor that generates a model that
specifies the man-in-the-middle attack.  To the theory, we add an
axiom that states that when a person generates a symmetric key, that
person will use an uncompromised key to encrypt its first message.
The axiom makes it so that the adversary cannot decrypt the message
sent by the person.

The generated model can be given to {\cpsa}.  It infers that the only
way the door can decrypt the person's message is if the person
encrypted the message using the door's public key.  Once that
inference is made, the door can conclude that the person sent the
messages expected in a run of this protocol.

DoorSEP was constructed with a flaw for expository purposes.  The
protocol can be repaired by including the door's public key in the
signed content, by making the first message
$\enc{\enc{K,D}{\iv{P}}}{D}$.

Flawed protocols are often deployed, and may be embedded in widely
used devices before the flaws are understood.  This example shows that
such protocols can still achieve desired security goals when used in a
restricted context.  If the context can be modeled using trust axioms,
Logical Protocol Analysis can be used to check whether the goals are
in fact met in the context of use.

\begin{sloppypar}
\paragraph{Protocols and theories.}  Hence, security conclusions
require protocol analysis combined with other properties, which we
will assume are characterized axiomatically by a theory $\thyG$.  In
the DoorSEP case, it is generated by the trust axiom.  We also regard
a protocol $\Pi$ as determining an axiomatic theory $\Th\Pi$, namely
the theory of $\Pi$'s executions, as $\Pi$ runs possibly in the
presence of a malicious adversary.  The models of this theory are runs
of the protocol.  Thus, we would like to understand the joint models
of $\thyG\cup\Th\Pi$, where of course these theories may share
vocabulary.
\end{sloppypar}

\paragraph{The enrich-by-need strategy.}  Indeed, our approach is to
construct \emph{minimal models} in a \emph{homomorphism order}.  We
refer to these minimal models as \emph{shapes}~\cite{Guttman10}.  The
shapes show all of the minimal, essentially different things that can
happen subject to $\thyG\cup\Th\Pi$:  every execution contains
instances---meaning homomorphic images---of the shapes.  This is
useful to the security analyst who can inspect the minimal models and
appraise whether they are compatible with his needs.  The analyst can
do this even without being able to explicitly state the key security
goals.  In the case in which $\thyG=\emptyset$, so that only $\Th\Pi$
matters, generating these shapes is the central functionality of
{\cpsa}~\cite{cpsa16}.

We call this approach to security analysis \emph{enrich-by-need},
since we build homomorphism-minimal models by rising stepwise in the
homomorphism order, gradually generating them all.  {\cpsa} does so
using a ``authentication test'' method, which yields a compact,
uniform way to generate the set of minimal models of the protocol
theory~\cite{Guttman10,cpsatheory11}.

Indeed, a further advantage arises in the case where there is a finite set of
finite shapes.  In that case, we can summarize them in a sentence,
called a \emph{shape analysis sentence} constructed as the disjunction
of their \emph{diagrams}~\cite{Guttman14,Ramsdell12}.  The diagram of
a finite model is (roughly) the conjunction of the atomic formulas
true in it.  The shape analysis sentence is thus true in all of the
shapes.  Moreover, its syntactic form ensures that its truth will be
preserved by homomorphisms.  Thus, it will be true in \emph{all}
models of $\thyG\cup\Th\Pi$.  Indeed, no strictly stronger formula can
be true in all the models.  We regard the shape analysis as a security
goal achieved by $\thyG\cup\Th\Pi$.

Thus, finding a finite set of finite shapes determines a strongest
security goal that the system achieves.

We already have a special tool, called {\cpsa}~\cite{cpsa17}, that
computes the shapes and their sentences for a protocol $\Pi$ acting
alone.  It uses optimized algorithms that we have proved correct for
protocol analysis~\cite{Guttman10,cpsatheory11}.  Thus, we need to
extend it so that it can cooperate with another tool to adapt its
results to provide models of the whole theory $\thyG\cup\Th\Pi$.  We
effectively split $\Th\Pi$ into two parts, a hard part $T_h$ and an
easy part $T_e$.  Only {\cpsa} will handle the hard part.

Our strategy is to program Z3~\cite{DeMoura08} to look for minimal
models of $\thyG\cup T_e$ that extend a fragment of a model.  When the
resulting model $\mA$ contains additional behavior of $\Pi$, we return
to {\cpsa} to handle the hard part $T_h$, enriching $\mA$ with some
possible executions.  We then return these extensions to Z3.  If this
process terminates, we have a minimal joint model.  By iterating our
search, we obtain a covering set of minimal joint models.

The program that orchestrates this use of Z3 is called \emph{Razor}.
It adapts the ideas of an earlier program of the same
name~\cite{razor15}.

\paragraph{Contributions.}  This report has two goals.  First, we
define and justify the methods that the new Razor uses to drive Z3 to
generate homomorphism-minimal models of a given theory.  These
homomorphisms are not necessarily embeddings; that is, a homomorphism  to construct %
may map distinct values in its source model to the same value in its
target model.  %
To begin with, we need a method to construct, from a model $\mA$, a set
of sentences $\avoid{\mA}$, true in precisely those models \mB such
that there is no homomorphism from $\mA$ to $\mB$.  %
We also need a method to construct, from a model \mA, a set of sentences
$\homTo{\mA}$, true in precisely those models \mB such that there
\emph{is} a homomorphism from $\mB$ to $\mA$.
We show how to use these two resources to compute
a set of minimal models that covers all of the models; this method is
codified in Razor.

Second, we develop a particular architecture for coordinating Razor
and {\cpsa}.  In this architecture, Razor handles all aspects of
$\thyG\cup\Th\Pi$ \emph{except} that it does not enrich a fragmentary
execution of $\Pi$ to obtain its shapes, i.e.~the minimal executions
that are its images.  Instead, we generate an input to {\cpsa} that
contains the substructure $\mA_0$ containing only protocol behavior.
{\cpsa} computes the shapes and extracts the strongest security goal
that applies to $\mA_0$.  It returns this additional information to
Razor, which then iterates.  We call this cooperative architecture
{\lpa} for \emph{Logical Protocol Analysis}.

\paragraph{Conclusions.}  We draw two main conclusions.  First, Razor
uses Z3 effectively to extract minimal models of a variety of
theories.  This is particularly true if the theories do not contain
many nested universal quantifiers.  Moreover, the {\lpa} coordination
between {\cpsa} and Razor is sound.

Second, when Razor and {\cpsa} are used together as in {\lpa}, Z3 must
handle theories with a fairly large number of nested universal
quantifiers.  Therefore this method is practical in its current form
only for quite small examples.  Refining the approach may enable us to
generate theories---possibly quantifier-free theories---that are
smaller and more easily digested by Z3.

\paragraph{Structure of this report.}  We organize the report into two
main chapters.  Chapter~\ref{sec:modulo:strands} introduces the
theories $\Th\Pi$, explains the way that {\lpa} marshals Razor and
{\cpsa} together.  Chapter~\ref{sec:finding minimal models} describes
Razor's strategies to use Z3 for finding minimal models, relative to a
given theory $T$.  Chapter~\ref{sec:concl} summarizes and concludes.

Within Chapter~\ref{sec:modulo:strands}, we introduce strand space
theory and give its axiomatic presentation in Section~\ref{sec:strand
  space theory}.  In Section~\ref{sec:doorsep}, we introduce an
example that uses Blanchet's Simple Example Protocol as a tool in an
authorization decision, namely whether to open a locked door.  The
protocol is chosen so that it would not necessarily be sound.
However, the analysis shows that the protocol is good enough given an
additional application-specific trust assumption.  This is the
assumption that, every time an authorized principal interacts with a
door, that door complies with the protocol and preserves the secrecy
of the values it is given.  This illustrates how protocol analysis may
be crafted to an application-specific context.

Within Chapter~\ref{sec:finding minimal models}, we lay the
foundations in Section~\ref{sec:foundations}, focusing on core models.
These are canonical homomorphism-minimal submodels, which have
embeddings into their homomorphic images.  We are, however, more
interested in homomorphisms that may not be embeddings; we introduce
this notion of minimality in Section~\ref{sec:minimality}.  This
section shows how to compute minimality models in either the embedding
sense or the sense of all homomorphisms.  We then turn from the theory
to the implementation considerations of working with the SMT solver Z3
or other SMT2-lib-compliant solvers.  Section~\ref{sec:results} gives
numerical results for DoorSEP and some other small examples.

%% file: z3cpsa.tex
\section{Model Finding Modulo Strand Space Theory}
\label{sec:modulo:strands}

This chapter shows how to use the model finders presented in
Chapter~\ref{sec:finding minimal models} and Strand
Spaces~\cite{ThayerHerzogGuttman99} to analyze cryptographic
protocols.  The implications of trust policies expressed in
first-order logic can be analyzed by studying their models.  A trust
policy that includes a theory about a cryptographic protocol allows
one to determine the impact of the policy on the execution of a
protocol.  However, deducing protocol executions is not something that
can be efficiently done within an {\smt} solver.  An external, finely
tuned tool is called for.

The Cryptographic Protocol Shapes Analyzer~\cite{cpsa17} (\cpsa) is a
tool that can be used to determine if a protocol achieves
authentication and secrecy goals.  It performs symbolic cryptographic
analysis based on the Dolev-Yao adversary model~\cite{DolevYao83} and
Strand Spaces.  Determining if a protocol satisfies a goal is an
undecidable problem, however, {\cpsa} appears to have a performance
advantage over other tools by using forward reasoning based on solving
authentication tests~\cite{GuttmanThayer02}.

{\cpsa} begins an analysis with a description of a protocol and an
initial scenario.  The initial scenario is a partial description of
the execution of a protocol.  If {\cpsa} terminates, it computes a
description of all of the executions of the protocol consistent with
the initial scenario.  For example, if it is assumed that one role of
a protocol runs to completion and {\cpsa} terminates, {\cpsa} will
determine what other roles must have executed.

Associated with each {\cpsa} protocol~$\Pi$ is a first-order
language~$\glang(\Pi)$ used to specify security
goals~\cite{Guttman14}.  The language can be used to exchange
information between {\cpsa} and an {\smt} solver.

A security goal is a sentence with a special form.  It is a
universally quantified implication.  Its hypothesis is a conjunction
of atomic formulas.  Its conclusion is a disjunction of existentially
quantified conjunctions of atomic formulas.  Security goals can be
used to express authentication and secrecy goals.

{\cpsa} describes a set of executions with an object called a
skeleton.  A skeleton that explicitly describes all of the
non-adversarial behavior in each execution is called a realized
skeleton.  Skeletons are presented in Section~\ref{sec:strand spaces}.

A Tarski style semantics, one that uses a skeleton as a model for a
sentence, is defined for each goal language.  A security goal is
achieved by a protocol if every realized skeleton models the goal.
The goal language is presented in Section~\ref{sec:protocol formulas}.

For use with {\smt} solvers, there is a theory~$T_\Pi$ for
protocol~$\Pi$.  When this theory is included, models restricted
to~$\glang(\Pi)$ characterize a skeleton of~$\Pi$.  The theory
associated with a protocol is presented in Section~\ref{sec:skeleton
  axioms}.

The goal language used by {\lpa} is strand-oriented as opposed to
being node-oriented.  The distinction is presented in
Section~\ref{sec:strand vs. node} along with the motivation for
choosing a strand-oriented language.

There is a special security goal that can be extracted from the
results of a run of {\cpsa}.  A Shape Analysis Sentence
(\sas)~\cite{Ramsdell12} encodes everything that has been learned
about the protocol from a {\cpsa} analysis starting with a given
initial scenario.  A {\sas} is used to import the results of a {\cpsa}
analysis into the {\smt} solver.

\begin{figure}
  \begin{center}
    \includegraphics{arch-0.mps}
  \end{center}
  \caption{{\lpa} Architecture}\label{fig:arch}
\end{figure}

The architecture for combining an {\smt} solver with {\cpsa}, called
the Logical Protocol Analyzer (\lpa), is displayed in
Figure~\ref{fig:arch}.  Theories are expressed using \smtlib~2.5
syntax.  An analysis begins with a {\cpsa} protocol~$\Pi$ and an
initial theory~$T_0$.  The initial theory contains a specification of
the trust policy and a description of the initial scenario of the
protocol as a collection of sentences in~$\glang^+(\Pi)$, an
extension of~$\glang(\Pi)$.

The program \texttt{prot2smt2} extracts the protocol theory~$T_\Pi$.
The initial theory is appended to the protocol theory to form the
first theory~$T_1$ to be analyzed by Razor, the model finder.  A
skeleton is extracted from each model.  If the skeleton is realized,
the model describes the impact of the trust policy on complete
executions of the protocol.  If the skeleton is not realized, it is
used as the initial scenario for {\cpsa}.  The results of {\cpsa} is
turned into a {\sas} and added to the current theory for further
analysis.  The {\lpa} algorithm is presented in Section~\ref{sec:lpa}.
An example of the use of {\lpa} is in Section~\ref{sec:doorsep}.

The performance of {\lpa} is not good.  Model finding consumes a large
amount of \textsc{cpu} time, even for small problems.
Section~\ref{sec:results} presents the results from running our test
suite and an analysis of {\lpa}'s performance issues.

%% file: solpa.tex
\subsection{Strand Space Theory}\label{sec:strand space theory}

This section describes the theory behind the strand-oriented
implementation of~\lpa. Unlike~\cite{Guttman14}, this paper uses
many-sorted first order logic, which is a better match for existing
software tools.  This paper incorporates much from~\cite{Ramsdell12}.
Unlike~\cite{Ramsdell12}, this paper uses one-based sequence indexing.

\paragraph{Notation.}

A finite sequence is a function from an initial segment of the
positive integers (\pos).  The length of a sequence~$X$ is~$|X|$, and
sometimes we write sequence~$X=\seq{X(1),\ldots, X(n)}$ for $n=|X|$.
The prefix of sequence~$X$ of length~$n$ is~$X\prefix n$.

\subsubsection{Message Algebras}\label{sec:alg}

\begin{figure}[!ht]
\[\begin{array}{ll@{{}\typ{}}ll}
\mbox{Sorts:}&\multicolumn{3}{l}{\mbox{$\srt{M}$, $\srt{T}$,
    $\srt{S}$, $\srt{A}$}}\\
\mbox{Functions:}&(\cdot,\cdot)&\srt{M}\times\srt{M}\to\srt{M}
&\mbox{Pairing}\\[1ex]
&\enc{\cdot}{(\cdot)}^S&\srt{M}\times\srt{S}\to\srt{M}
&\mbox{Symmetric encryption}\\[1ex]
&\enc{\cdot}{(\cdot)}^A&\srt{M}\times\srt{A}\to\srt{M}
&\mbox{Asymmetric encryption}\\[1ex] &(\cdot)^{-1}&\srt{A}\to\srt{A}
&\mbox{Asymmetric key inverse}\\ &\cnc{tt}&\srt{T}\to\srt{M}
&\mbox{Text inclusion}\\ &\cnc{sk}&\srt{S}\to\srt{M} &\mbox{Symmetric
  key inclusion}\\ &\cnc{ak}&\srt{A}\to\srt{M} &\mbox{Asymmetric key
  inclusion}\\
\mbox{Equation:}&\multicolumn{3}{l}{(x^{-1})^{-1}=x\mbox{ for $x:\srt{A}$}}
\end{array}\]
\caption{Simple Crypto Algebra Signature}\label{fig:signature}
\end{figure}

Figure~\ref{fig:signature} shows the simplification of the {\cpsa}
message algebra signature used by {\lpa}.  Sort~$\srt{M}$ is the sort
of messages.  The other sorts, sort~$\srt{A}$ (asymmetric keys),
sort~$\srt{S}$ (symmetric keys), and sort~$\srt{T}$ (text), are called
\emph{basic sorts}.

A message constructed by applying a term of a basic sort to an
inclusion function is call a \emph{basic value}.  Messages are
generated from the basic values using encryption
$\enc{\cdot}{(\cdot)}^{\{S,A\}}$ and pairing $(\cdot,\cdot)$, where
the comma function is right associative and parentheses are
omitted when the context permits.

A set of variables~$X$ is \emph{well-sorted} if for each~$x\in X$, $x$
has a unique sort~$S$.  Let $x\typ S$ assert that the sort of~$x$
is~$S$ in~$X$.  Suppose~$X$ and~$Y$ are well-sorted and contain~$x$.
The sort of~$x$ in~$X$ need not agree with the sort of~$x$ in~$Y$.

Let~$\alga(X)$ be the quotient term algebra generated by a
set of well-sorted variables~$X$.  We often leave the set of variables
implicit, and refer to the carrier set for sort~$S$ by $\alga_S$.

A message~$t_1$ is \emph{carried by}~$t_2$, written $t_1\sqsubseteq
t_2$ if~$t_1$ can be derived from~$t_2$ given the right set of keys,
that is $\sqsubseteq$ is the smallest reflexive, transitive relation
such that $t_1\sqsubseteq t_1$, $t_1\sqsubseteq (t_1, t_2)$,
$t_2\sqsubseteq (t_1, t_2)$, and $t_1\sqsubseteq\enc{t_1}{t_2}$.

\begin{figure}
  \[\begin{array}{r@{{}\vdash{}}l@{\qquad}r@{{}\vdash{}}l}
  t&t&t_1,\cn{sk}(t_2)&\enc{t_1}{t_2}^S\\
  t_1,t_2&(t_1,t_2)&\enc{t_1}{t_2}^S,\cn{sk}(t_2)&t_1\\
  (t_1,t_2)&t_1&t_1,\cn{ak}(t_2)&\enc{t_1}{t_2}^A\\
  (t_1,t_2)&t_2&\enc{t_1}{t_2}^S,\cn{ak}(t_2^{-1})&t_1
  \end{array}\]
  \caption{Adversary Derivability Relation}\label{fig:derivability}
\end{figure}

Adversary behavior is modeled via a derivability relation.  Given a
set of messages~$S$, message~$t$ is derivable, written~$S\vdash t$,
when there is a derivation using the rules in
Figure~\ref{fig:derivability}.  These rules encode the Dolev-Yao
model~\cite{DolevYao83}.

\subsubsection{Strand Spaces}\label{sec:strand spaces}

A run of a protocol is viewed as an exchange of messages by a finite
set of local sessions of the protocol.  Each local session is called a
strand.  A \emph{strand} is a finite sequence of events.  An \emph{event} is
either a message transmission or a reception.  Outbound message
$t\in\alga_\srt{M}$ is written as~$\outbnd t$, and inbound message~$t$
is written as~$\inbnd t$.   A message \emph{originates} in a strand
if it is carried by some event and the first event in which it is
carried is outbound.  A message is \emph{acquired} in a
strand if it is carried by some event and the first event in which it is
carried is inbound.

A \emph{strand space}~$\Theta$ is a finite sequence of strands.  A
message that originates in exactly one strand of~$\Theta$ is
\emph{uniquely originating}, and represents a freshly chosen value.  A
message is \emph{mentioned} in~$\Theta$ if it occurs in a strand
of~$\Theta$, or if it is an asymmetric key, its inverse occurs in a
strand of~$\Theta$.  A message that is mentioned but originates
nowhere in~$\Theta$ is \emph{non-originating}, and often represents an
uncompromised key.

A node identifies an event in a strand space.  A \emph{node} is a pair
of positive integers, and the event associated with node $(s,i)$ is
$\evt_\Theta(s,i) =\Theta(s)(i)$.  We sometimes omit the strand space
when it is obvious from the context.  The set of nodes of strand
space~$\Theta$ is
\[\nodes(\Theta)=\{(s,i)\mid s\in\sdom(\Theta),
i\in\sdom(\Theta(s))\}.\]

The \emph{strand succession} relation is the binary relation
${}\Rightarrow{}\colon\nodes(\Theta)\times\nodes(\Theta)$, such that
\[(s_1,i_1)\Rightarrow(s_2,i_2)\mbox{ iff }s_1=s_2\mbox{ and }i_1+1=i_2.\]

An execution is called a bundle.  A \emph{bundle}
$\bundle(\Theta,\to)$ is a strand space~$\Theta$ and a binary
communication relation
${}\to{}\colon\nodes(\Theta)\times\nodes(\Theta)$, such that
\begin{enumerate}
\item the graph with $\nodes(\Theta)$ as vertices and
  $\Rightarrow\cup\to$ as edges is acyclic;
\item if $n_0\rightarrow n_1$, then $\evt(n_0)=\outbnd t$
  and~$\evt(n_1)=\inbnd t$ for some~$t$; and
\item for each reception node~$n_1$, there is a unique transmission
  node~$n_0$ with $n_0\rightarrow n_1$.
\end{enumerate}
The node precedence relation of~\bundle,
${\prec_\bundle}=(\Rightarrow\cup\to)^+$, is a strict partial ordering
of nodes and represents the causal relation between events that occur
at nodes in~\bundle.  In a bundle, a strand that is an instance of a
role in Figure~\ref{fig:pen} is called a \emph{penetrator strand},
and the remaining strands are \emph{regular}.  In what follows, we
assume all regular strands precede penetrator strands in the sequence
of strands~$\Theta$, that is, if $\Theta_s$ is regular and
$\Theta_{s'}$ is a penetrator strand, then $s<s'$.

\begin{figure}
$$\begin{array}{lll}
\seq{\inbnd x,\inbnd y,\outbnd(x, y)}&
\seq{\inbnd(x, y),\outbnd x,\outbnd y}&
\mbox{Pair}\\
\seq{\inbnd x,\inbnd\cn{sk}(y),\outbnd\enc{x}{y}}&
\seq{\inbnd\enc{x}{y},\inbnd \cn{sk}(y),\outbnd x}&
\mbox{Symmetric}\\
\seq{\inbnd x,\inbnd\cn{ak}(y),\outbnd\enc{x}{y}}&
\seq{\inbnd\enc{x}{y},\inbnd \cn{ak}(y^{-1}),\outbnd x}&
\mbox{Asymmetric}\\
\seq{\outbnd\cn{tt}(x)}\quad
\seq{\outbnd\cn{sk}(x)}&
\seq{\outbnd\cn{ak}(x)}&
\mbox{Create}
\end{array}$$
\caption{Penetrator Roles}\label{fig:pen}
\end{figure}

A skeleton represents all or part of the regular portion of an
execution.  A \emph{skeleton} $k=\skel_X(\Theta,\prec,\nu,\upsilon)$,
where~$X$ is a set of well-sorted variables used to generate the
message algebra used by~$\Theta$,~$\prec$ is a strict partial ordering
of the nodes in $\Theta$,~$\nu$ is a set of basic values mentioned
in~$\Theta$, none of which is carried in a strand in~$\Theta$,
and~$\upsilon$ is a set of pairs consisting of a basic value and a
node.  For each $(t,n)\in\upsilon$, $t$ originates at~$n$ in $\Theta$
and at no other node.  In addition,~$\prec$ must order the node for
each event that receives a uniquely originating basic value after the
node of its transmission, so as to model the idea that the basic value
represents a value freshly generated when it is transmitted.

Skeleton $k=\skel_X(\Theta,\prec,\nu,\upsilon)$ is the \emph{skeleton
  of} bundle $\bundle(\Theta',\to)$ if
\begin{enumerate}
\item $\Theta=\Theta'\prefix n$, where $n$ is the number of regular
  strands in $\Theta'$;
\item $\prec$ is the restriction of $\prec_\bundle$ to
  $\nodes(\Theta)$;
\item $\nu$ is the set of non-originating basic values in $\Theta$;
  and
\item $\upsilon$ is the set of uniquely originating basic values and
  their node of origination in $\Theta$.
\end{enumerate}

Let $k=\skel_X(\Theta,\prec,\nu,\upsilon)$ and
$k'=\skel_{X'}(\Theta',\prec',\nu',\upsilon')$ be skeletons.  There is
a \emph{skeleton homomorphism} $(\phi,\sigma)\colon k\mapsto k'$
if~$\phi$ and~$\sigma$ are maps with the following properties:
\begin{enumerate}
\item $\phi$ maps strands of~$k$ into those of~$k'$, and nodes as
  $\phi(s,i)=(\phi(s),i)$, that is $\phi$ is in
  $\sdom(\Theta)\to\sdom(\Theta')$;
\item $\sigma\colon\alga(X)\to\alga(X')$ is a message algebra homomorphism;
\item $n\in\nodes(\Theta)$ implies
  $\sigma(\evt_\Theta(n))=\evt_{\Theta'}(\phi(n))$;
\item $n_0\prec n_1$ implies $\phi(n_0)\prec'\phi(n_1)$;
\item $\sigma(\nu)\subseteq \nu'$;
\item $(t,n)\in\upsilon$ implies $(\sigma(t),\phi(n))\in\upsilon'$.
\end{enumerate}
Skeleton~$k$ \emph{covers} bundle~{\bundle} if there exists a
homomorphism from~$k$ to the skeleton of~\bundle.  Skeleton~$k$ is
\emph{realized} iff there is an injective homomorphism to the
skeleton of some bundle that preserves the length of strands.

A protocol~$\Pi$ is a strand space with restrictions and a theory.
The details of the theory of~$\Pi$, written~$T_\Pi$, will be presented
in the next section.  A strand~$i$ of~$\Pi$, written $\Pi_i$ is called
a role.  A \emph{role} is a strand where every variable of sort
message that occurs in the strand is acquired.  Strand~$s$ is an
\emph{instance} of role~$\Pi_i$ if~$s$ is a prefix of the result of
applying some substitution~$\sigma$ to~$\Pi_i$.  Every protocol
contains the listener role $\seq{\inbnd x,\outbnd x}$ for
$x\typ\srt{M}$.

Skeleton~$k$ is a skeleton of protocol~$\Pi$ if each strand
in~$\Theta$ is an instance of a protocol role, and~$k$ models
theory~$T_\Pi$ as defined in the next section.

\subsubsection{Protocol Formulas}\label{sec:protocol formulas}

The signature~$\lang(\Pi)$ used for strand-oriented protocol (\sop)
formulas includes of the sorts and functions in the underlying message
algebra.  There are two additional sorts:~$\srt{D}$, the sort for
strands and~\srt{I}, the sort for indices.  A node is never explicitly
represented, instead it is represented as a pair consisting of a
strand and an index.

{\sop} formulas make use of protocol specific and protocol independent
predicates.  For each role $\Pi_i$, there is a protocol specific
binary strand length predicate $\Pi_i:\srt{D}\times\srt{I}$.  For each
role~$\Pi_i$ and variable~$x:S$ that occurs in $\Pi_i$, there is a
protocol specific binary strand parameter predicate
$\Pi_i^x:\srt{D}\times S$.  The protocol independent unary predicates
are $\mathsf{non}B: B$ for each basic sort~$B\in\{\srt{T},
\srt{S},\srt{A}\}$.  The ternary protocol independent predicates are
$\mathsf{uniqAt}B: B\times\srt{D}\times\srt{I}$.  The quaternary
protocol independent predicate is
$\mathsf{prec}:\srt{D}\times\srt{I}\times\srt{D}\times\srt{I}$.
Equality is part of the signature.  Finally, there is a constant of
sort~\srt{I} for each index in the longest role of the protocol.
Thus, suppose the longest role has length~$n$, then $1:\srt{I},
2:\srt{I},\ldots,n:\srt{I}$ are constants in the signature.

\paragraph{Semantics of Protocol Formulas.}

Let $k=\skel_X(\Theta,\prec,\nu,\upsilon)$.  The universe of
discourse~$\alg{D}$ contains a set for each sort in $\lang(\Pi)$.  For
sort~\srt{D}, $\alg{D}_\srt{D}$ is the domain of~$\Theta$.
$\alg{D}_\srt{I}$ is the set of integers that correspond to the
constants of sort~\srt{I}.  The universe of discourse~$\alg{D}_S$ for
each algebra sort~$S$ is~$\alga(X)_S$.

When formula~$\Phi$ is satisfied in skeleton~$k$ with variable
assignment $\alpha\typ Y\to \alg{D}$, we write
$k,\alpha\kmdls\Phi$.  We write~$\bar\alpha$ when~$\alpha$ is
extended to terms in the obvious way.  When sentence~$\Gamma$ is
modeled by skeleton~$k$, we write $k\kmdls\Gamma$.

\begin{itemize}
\item $k,\alpha\kmdls\Pi_i(y,j)$ iff for some~$s=\alpha(y)$
  and~$\sigma$, $s$ is in the domain of~$\Theta$, and
  \[\Theta_s\prefix\bar\alpha(j)=\sigma(\Pi_i\prefix\bar\alpha(j)).\]

\item $k,\alpha\kmdls\Pi^x_i(y,t)$ iff
   for some~$s=\alpha(y)$ and~$\sigma$,
  $s$ is in the domain of~$\Theta$,
  $x$ first occurs in $\Pi_i$ at $j$, and for some~$\sigma$ with
  $\sigma(x)=\bar\alpha(t)$,
  \[\Theta_s\prefix j=\sigma(\Pi_i\prefix j).\]
\end{itemize}

The interpretation of the protocol independent predicates is
straightforward.
\begin{itemize}
\item $k,\alpha\kmdls\cn{prec}(y_d,y_i,z_d,z_i)$ iff
  $(\alpha(y_d),\bar\alpha(y_i))\prec(\alpha(z_d),\bar\alpha(z_i))$.
\item $k,\alpha\kmdls\cn{nonT}(t)$ iff $\cn{tt}(\bar\alpha(t))\in\nu$.
\item $k,\alpha\kmdls\cn{nonS}(t)$ iff $\cn{sk}(\bar\alpha(t))\in\nu$.
\item $k,\alpha\kmdls\cn{nonA}(t)$ iff $\cn{ak}(\bar\alpha(t))\in\nu$.
\item $k,\alpha\kmdls\cn{uniqAtT}(t,y_d,y_i)$ iff
  $(\cn{tt}(\bar\alpha(t)),(\alpha(y_d),\bar\alpha(y_i)))\in\upsilon$.
\item $k,\alpha\kmdls\cn{uniqAtS}(t,y_d,y_i)$ iff
  $(\cn{sk}(\bar\alpha(t)),(\alpha(y_d),\bar\alpha(y_i)))\in\upsilon$.
\item $k,\alpha\kmdls\cn{uniqAtA}(t,y_d,y_i)$ iff
  $(\cn{ak}(\bar\alpha(t)),(\alpha(y_d),\bar\alpha(y_i)))\in\upsilon$.
\item $k,\alpha\kmdls y=z$ iff $\bar\alpha(y)=\bar\alpha(z)$.
\end{itemize}

\paragraph{Associated Protocol Theories.}

In addition to a sequence of roles, protocol~$\Pi$ has an associated
theory~$T_\Pi$.  For example, the axiom
\[\all{x\typ\srt{D},b\typ\srt{T}}\Pi_1(x,1)\land\Pi_1^n(x,b)
\supset\cn{uniqAtT}(b,x,1)\] states that the~$n$ parameter of an
instance of role~$\Pi_1$ always uniquely originates at its first node.

\subsubsection{Logical Protocol Analysis}\label{sec:lpa}


The \emph{protocol theory} of protocol~$\Pi$, $\thyr{Tr}(\Pi)$ is:
\[\thyr{Tr}(\Pi)=\Upsilon(\Pi)\cup\{\Gamma\mid\mbox{for all realized
  $k$, $k\kmdls\Gamma$}\},\]
where $\Upsilon(\Pi)$ is the theory of $\Pi$-skeletons presented in
Section~\ref{sec:skeleton axioms}.

For each protocol~$\Pi$, ideally one would like to perform model
finding modulo~$\thyr{Th}(\Pi)$.  The {\lpa} program performs our
approximation.

The inputs to {\lpa} are a protocol~$\Pi$ and an initial theory~$T_0$.
The signature of the initial theory, $\bar\lang(\Pi)$ may extend the
signature of protocol formulas with new sorts, functions, and
predicates.  The additional functions and predicates are treated as
uninterpreted symbols.

Suppose~$\mdlm$ models~$\Gamma$ ($\mdlm\models\Gamma$).  We denote by
$\mdlm\Downarrow\lang(\Pi)$ the \emph{reduct} of~$\mdlm$ to $\lang(\Pi)$
by restricting it to interpret only symbols in~$\lang(\Pi)$.  Assume
each element in the domain of model~$\mdlm$ is distinct from a function
in~$\bar\lang(\Pi)$.  Let $\bar\lang(\Pi)^+$ be the extension of
$\bar\lang(\Pi)$ in which each element in the model is a constant of
the appropriate sort.  The \emph{positive diagram} of~$\mdlm$,
$\diag(\mdlm)$, is the set of all of the atomic $\bar\lang(\Pi)^+$
sentences modeled by~$\mdlm$.

Given a set of atomic $\lang(\Pi)^+$ $\Phi$ sentences,
Guttman~\cite[Section~4.3]{Guttman14} describes when and how one can
extract a skeleton $k=\fn{cs}(\Phi)$ that is characterized by~$\Phi$.

{\lpa} starts with theory $T_1$.  It is the union of initial
theory~$T_0$, $T_\Pi$, and $\Upsilon(\Pi)$, the $\Pi$-skeleton
theory presented in Section~\ref{sec:skeleton axioms}.

\begin{algo}\label{alg:lpa}
{\lpa} performs the following procedure to output a set of models.
\vskip2ex
\noindent $\lpa(T){}\equiv{}$
\begin{enumerate}
\item Let $\mdlm_i$ be a set of minimal models of~$T$.  If
  $T$ is unsatisfiable, abort.

\item For each $\mdlm_i$ do:
\begin{enumerate}
\item $\Gamma\gets\diag(\mdlm_i\Downarrow\lang(\Pi))$.  If $\Gamma$ is
  not role specific, signal initial theory error and abort.

\item If $\fn{cs}(\Gamma)$ is realized, output $\mdlm_i$.

\item Use $\fn{cs}(\Gamma)$ as the point-of-view for {\cpsa} and compute the
  resulting shape analysis sentence~$\Delta$ using the algorithm
  in~\cite{Ramsdell12}.

\item $\lpa(T\cup\{\Delta\})$.
\end{enumerate}
\end{enumerate}
\end{algo}

\subsubsection{Skeleton Axioms}\label{sec:skeleton axioms}

This section presents the $\Pi$-skeleton theory, $\Upsilon(\Pi)$, used
to form an initial theory by {\lpa}.  For performance reasons, the
signature used omits the message sort~$\srt{M}$ along with all message
algebra functions that refer to that sort.  The only message algebra
function that remains is asymmetric key inverse~$(\cdot)^{-1}$.  This
simplification of the protocol signature was necessary because model
finders attempt to create an infinite domain for the message sort
using the obvious specification.

\begin{sloppypar}
There are three new binary predicates introduced in this section,
predicates $\mathsf{carriedAt}B: B\times\srt{D}\times\srt{I}$, for
each basic sort~$B\in\{\srt{T}, \srt{S},\srt{A}\}$.  Formula
$\mathsf{carriedAtT}(b,x,i)$ asserts that text~$b$ is first carried in
strand~$x$ by the message at position~$i$.  Finally, there is an index
precedence predicate ${\iprec}\typ\srt{I}\times\srt{I}$.
\end{sloppypar}

In what follows, $\cn{prec}(\cdot,\cdot,\cdot,\cdot)$ is written as
$(\cdot,\cdot)\prec(\cdot,\cdot)$.

\begin{enumerate}
\item All message functions are injective.
  \[\all{mn\typ\srt{A}}m^{-1}=n^{-1}\supset m=n.\]

\item For $k\typ\srt{A}$, $k^{-1}\neq k$ and $(k^{-1})^{-1}=k$.

\item Index constants enumerate the sort.  For all $1\leq i,j \leq n$,
  \[\all{x\typ\srt{I}}\some{i}x=\cnc{c}_i\mbox{ and
  $\cnc{c}_i\neq\cnc{c}_j$ when $i\neq j$.}\]

\item Index constants are linearly ordered.  For all $1\leq i,j \leq n$,
  \[\cnc{c}_i\iprec\cnc{c}_j\mbox{ iff }i<j.\]

\item Uniquely originating values originate at at most one node:
  \[\begin{array}{c}
  \all{v\typ S,xy\typ\srt{D},ij\typ\srt{I}}\mathsf{uniqAt}S(v,x,i)
  \land\mathsf{uniqAt}S(v,y,j)\\
  {}\supset x=y\land i=j,
  \end{array}\] for $S\in\{\srt{T},\srt{S},\srt{A}\}$.

\item Carried implies not non-origination:
  \[\all{v\typ S,x\typ\srt{D},i\typ\srt{I}}\mathsf{carriedAt}S(v,x,i)\land
  \mathsf{non}S(v)\supset\falsehood,\]
  for $S\in\{\srt{T},\srt{S},\srt{A}\}$.

\item Precedence is strict:
  \[
  \begin{array}{c}
    \all{x\typ\srt{D},i\typ\srt{I}}(x,i)\not\prec (x,i)\\
    \all{xyz\typ\srt{D},ijk\typ\srt{I}}(x,i)\prec (y,j)\land (y,j)\prec (z,k)
    \supset (x,i)\prec (z,k).
  \end{array}\]

\item Strand succession rule:
  \[\begin{array}{c}
  \all{x\typ\srt{D},ijk\typ\srt{I}}(x,i)\prec(x,k)\land i\iprec
  j\land j\iprec z\\
  {}\supset (x,i)\prec(x,j)\land(x,j)\prec(x,k).
  \end{array}
  \]

\item Node of unique origination is before carried node:
  \[\begin{array}{c}
  \all{v\typ S,x,y\typ\srt{D},i,j\typs{I}}\mathsf{uniqAt}S(v,x,i)\land
  \mathsf{carriedAt}S(v,y,j)\\
  {}\supset (x=y\land i=j)\lor (x,i)\prec (y,j),
  \end{array}\]
  for $S\in\{\srt{T},\srt{S},\srt{A}\}$.

\item Strand predecessor node exists for~$\Pi_i$:
  \[\all{x\typ\srt{D},hj\typ\srt{I}}\Pi_i(x,h)\land j\iprec h\supset
  (x,j)\prec(x,h).\]

\item Non-overlapping role pairs respect node positions:
  \begin{itemize}
  \item If $\Pi_i\prefix 1$ does not unify with $\Pi_j\prefix 1$,
  \[\all{x\typ\srt{D},hk\typ\srt{I}}
  \Pi_i(x,h)\land \Pi_j(x,k)\supset\falsehood.\]
  \item If $\ell$ is the largest position such that
  $\Pi_i\prefix\ell$ unifies with $\Pi_j\prefix\ell$,
  \[\all{x\typ\srt{D},hk\typ\srt{I}}
  \Pi_i(x,h)\land \Pi_j(x,k)\land\cnc{c}_\ell\iprec h\land\cnc{c}_\ell\iprec k
  \supset\falsehood.\]
  \end{itemize}

\item Parameter first occurrence node:
  \begin{itemize}
  \item If $v\typ S$ first occurs in~$\Pi_i$ at~1,
    \[\all{x\typ\srt{D},h\typs{I}}\Pi_i(x,h)\supset
    \some{m\typ S}\Pi_i^v(x,m).\]
  \item If $v\typ S$ first occurs in~$\Pi_i$ at~$\ell$, where $\ell>1$,
    \[\all{x\typ\srt{D},h\typs{I}}\Pi_i(x,h)\land\cnc{c}_{\ell-1}\iprec
    h\supset\some{m\typ S}\Pi_i^v(x,m).\]
  \end{itemize}

\item There is at most one value for each parameter:
\[\all{x\typ\srt{D}, m n\typ S}
\Pi_i^v(x, m)\land \Pi_i^v(x, n)\supset m=n.\]

\item Assert first carried node: \\
  Assume basic value $b\typ S$ is first carried in~$\Pi_i$ at~$\ell$, $v$ is
  the variable that occurs in~$b$, and $n=m$ if $b=v$ else $n=m^{-1}$.
  \begin{itemize}
  \item If $\ell=1$,
    \[\all{x\typ\srt{D},h\typs{I},m\typ S}\Pi_i(x,h)\land\Pi_i^v(x,n)
    \supset\mathsf{carriedAt}(n, x,\cnc{c}_1).\]
  \item If $\ell>1$,
    \[\begin{array}{c}
    \all{x\typ\srt{D},h\typs{I},m\typ
      S}\Pi_i(x,h)\land\cnc{c}_{\ell-1}\iprec h\land\Pi_i^v(x,n)\\
        {}\supset\mathsf{carriedAt}(n, x,\cnc{c}_\ell).
    \end{array}\]
  \end{itemize}

\end{enumerate}

Let $\Upsilon(\Pi)$ be the collection of axioms listed above in this
section.  Observe that $\Upsilon(\Pi)$ is geometric.  Let $\Gamma$ be
a $\bar\lang(\Pi)$ sentence.  Model~$\mdlm$ is a $k$ \emph{skeleton
  model} of~$\Gamma$, written $\mdlm\models_k\Gamma$, iff
\[\mdlm\models\Gamma\land\Upsilon(\Pi)\mbox{ and }
k=\fn{cs}(\diag(\mdlm\Downarrow\lang(\Pi))). \]

\subsubsection{Strand-Oriented vs.\@ Node-Oriented Goal
  Languages}\label{sec:strand vs. node}

A goal language is \emph{strand-oriented} if non-algebra logical
variables denote strands.  A goal language is \emph{node-oriented} if
non-algebra logical variables denote nodes.

In the original work on goal languages~\cite{Guttman14}, all languages
are node-oriented.  Node-oriented languages are to be preferred when
considering protocol transformations, that is, when mapping one
protocol into a larger protocol, and ensuring the goals in the source
protocol are properly reflected in the target protocol.  {\cpsa}
version~3 \cite{cpsa16} provides a node-oriented goal language.  The
first implementation of {\lpa} used {\cpsa}3 and a node-oriented goal
language.

The goal language presented is Section~\ref{sec:protocol formulas} is
strand-oriented.  Strand-oriented languages have the advantage that
they are more compact and easier to understand.  Furthermore, they
align more naturally with {\cpsa} input and output, which itself is
strand-oriented.  {\cpsa} version~4 \cite{cpsa17} provides a
strand-oriented goal language.  The reason {\cpsa}4 came into
existence is not due to positive features of strand-oriented
languages, but instead it was to mitigate the {\lpa} performance issues.
Strand-oriented models of skeletons produced by Z3 are smaller than
node-oriented models.  Section~\ref{sec:results} contains an example
in which a strand-oriented analysis is sixteen times faster that a
node-oriented analysis.

%% file: doorsep.tex
\subsection{Door Simple Example Protocol}\label{sec:doorsep}

We present the DoorSEP Protocol example, and its analysis in full
detail based on the foundation presented in the previous section.  In
this example, imagine there is a door with a badge reader, and a
person with a badge.  The door has opened.  We want to know what else
must have happened.

Recall the diagram in Figure~\ref{fig:arch} to visualize the analysis
process.

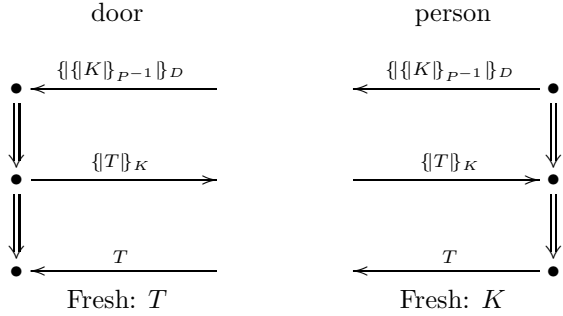
\begin{figure}
  \[\begin{array}[c]{c@{\qquad\qquad}c}
  \mbox{door}&\mbox{person}\\[2ex]
  \xymatrix@C=7em{
    \bullet\ar@{=>}[d]
    &\ar[l]_{\enc{\enc{K}{\iv{P}}}{D}}\\
    \bullet\ar@{=>}[d]\ar[r]^{\enc{T}{K}}&\\
    \bullet&\ar[l]_T}&
  \xymatrix@C=7em{
    &\bullet\ar@{=>}[d]\ar[l]_{\enc{\enc{K}{\iv{P}}}{D}}\\
    \ar[r]^{\enc{T}{K}}&\bullet\ar@{=>}[d]\\
    &\bullet\ar[l]_T}\\[2ex]
  \mbox{Fresh: $T$}&\mbox{Fresh: $K$}
  \end{array}\]
  \caption{DoorSEP Protocol}\label{fig:doorsep protocol}
\end{figure}

To begin this analysis, we must know how the person's badge was used
to authenticate.  We employ the protocol in Figure~\ref{fig:doorsep
  protocol}, which has been designed to have a flaw.  It is based on the
Simple Example Protocol due to Bruno Blanchet.  It fails to achieve
mutual authentication, however, we will add a trust axiom that
restores this property.

In this protocol, a person begins by generating a fresh symmetric
key, signing it, and then encrypting the result using the door's
public key.  If the door accepts the first message, it responds by
freshly generating a token and uses the symmetric key to encrypt it.
If the door receives the token back unencypted, the door concludes the
person that generated the key is at the door and opens.

The initial theory specifies the trust axiom and the fact that the
door is open.  To assert the door is open, one asserts there is a
strand that is a full length instance of the door role by declaring
constant $\cnc{s}\typ\srt{D}$ and asserting $\cn{door}(\cnc{s},
\cnc{c}_3)$.  We further assume a person's private key is
uncompromised by declaring constant $\cnc{p}\typ\srt{A},$ and
asserting $\cn{door}^p(\cnc{s}, \cnc{p})$ and
$\cn{nonA}(\iv{\cnc{p}})$.  The trust axiom is
\begin{equation}\label{eqn:trust axiom}
  \begin{array}{l}
    \all{p,d\typ\srt{A}, s\typ\srt{D}}\cn{nonA}(\iv{p})\land
    \cn{person}(s, \cnc{c}_1)\\
    \quad\land\cn{person}^p(s, p)
    \land\cn{person}^d(s, d)\supset\cn{nonA}(\iv{d}).
  \end{array}
\end{equation}
and will be explained later.

\begin{figure}
  \[\begin{array}{l}
  \some{i_0,i_1,i_2\typ\srt{I},s_0\typ\srt{D},a_0,a_1,a_2,a_3\typ\srt{A},
    k_0\typ\srt{S}, t_0\typ\srt{T}}\\
  i_0=\cnc{c}_1\land i_1=\cnc{c}_2\land i_2=\cnc{c}_3
  \land i_0\ll i_1\land i_0\ll i_2\land i_1\ll i_2\\
  {}\land a_1=\iv{a_0}\land a_0=\iv{a_1}\land a_3=\iv{a_2}\land a_2=\iv{a_3}
  \land s_0=\cnc{s}\land a_2=\cnc{p}\\
  {}\land\cn{door}(s_0,i_2)\land\cn{door}^t(s_0,t_0)
  \land\cn{door}^k(s_0,k_0)
  \land\cn{door}^p(s_0,a_2)
  \land\cn{door}^d(s_0,a_0)\\
  {}\land\cn{uniqAtT}(t_0,s_0,i_1)\land\cn{carriedAtT}(t_0, s_0, i_1)
  \land\cn{carriedAtS}(k_0, s_0,i_0)\\
  {}\land(s_0,i_0)\prec(s_0,i_1)
  \land(s_0,i_0)\prec(s_0,i_2)\land(s_0,i_1)\prec(s_0,i_2)\land\cn{nonA}(a_3)
  \end{array}\]
  \caption{DoorSEP First Model}\label{fig:doorsep first model}
\end{figure}

After appending the initial theory to the Skeleton Axioms for the
DoorSEP protocol specified in Section~\ref{sec:skeleton axioms}, Razor
finds the model in Figure~\ref{fig:doorsep first model}.  The first
line declares the set of elements in the domain.  The second line
shows the properties we depend on for integers.  The third line shows
that we have two pairs of asymmetric keys.  The forth line asserts we
have an instance of the door role of full length.  The remainder
asserts properties of basic values and the orderings implied by strand
succession.

\begin{figure}
  \[\xymatrix@C=7em{
    \txt{\strut door}&&\txt{\strut person}\\
    \bullet\ar@{=>}[d]
    &\succ\ar[l]_{\enc{\enc{K}{\iv{P}}}{D}}
    &\bullet\ar[l]_{\enc{\enc{K}{\iv{P}}}{D'}}\\
    \bullet\ar@{=>}[d]\ar[r]^{\enc{T}{K}}&\\
    \bullet&\ar[l]_T}\]
  \begin{center}
    Uncompromised: $P\quad$ Fresh: $K, T$
  \end{center}
  \caption{DoorSEP First Shape}\label{fig:doorsep first shape}
\end{figure}
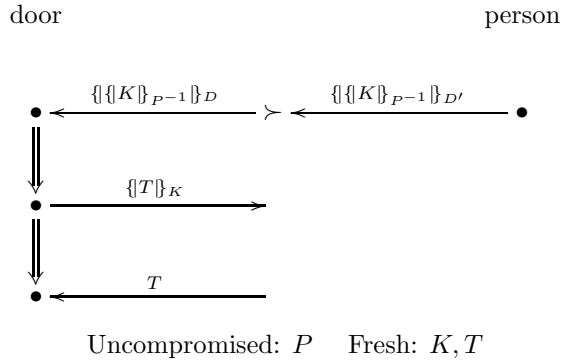

At this stage, we have a model that characterizes an unrealized
skeleton, and we would like to use {\cpsa} to find out what else must
have happened.  Ignoring atomic formulas using the predicates
\cn{carriedAtT}, \cn{carriedAtS}, and $\ll$ produces a model in the
language of the DoorSEP protocol.  The Skeleton Axioms ensure the
extraction of a skeleton that can be used as an initial {\cpsa}
scenario.  The shape produced by {\cpsa} is displayed in
Figure~\ref{fig:doorsep first shape}.

The shape shows the lack of mutual authentication
built into this flawed protocol.  To open the door, a person can use
an arbitrary compromised key for the door.  That is, without the trust
axiom, the answer to the what else happened question is that the
person holding private key~$\iv{\cnc{p}}$ swiped with their badge, but
the key used to identify the door may have been compromised, and an
adversary may have completed the rest of the protocol.

Consider the case in which the door is well known to the owner of the
badge.  For example, suppose the badge is issued by the institution
that owns the door and is tamper proof.  In that case, the person
knows to initiate the DoorSEP protocol (swipe their badge) only when
in front of a door with the correct key.  The trust axiom in
Eq.~\ref{eqn:trust axiom} codifies this policy.  It states that if a
person with an uncompromised key initiates the protocol, the door key
used is uncompromised.

\begin{figure}
  \[\begin{array}{l}
  \all{t_0\typ\srt{T},k_0\typ\srt{S},a_0,a_2\typ\srt{A},s_0\typ\srt{D}}\\
    \quad\cn{door}(s_0, \cnc{c}_3)\land
    \cn{door}^t(s_0, t_0)\land
    \cn{door}^k(s_0, k_0)\land
    \cn{door}^d(s_0, a_0)\\
    \quad\land\cn{door}^p(s_0, a_2)\land
    \cn{nonA}(\iv{a_2})\land
    \cn{uniqAtT}(t_0,s_0,\cnc{c}_2)\\
    \quad\supset\some{s_1\typ\srt{D},d\typ\srt{A}}\\
    \qquad\cn{person}(s_1,\cnc{c}_1)\land
    \cn{person}^k(s_1,k_0)\land
    \cn{person}^d(s_1,d_0)\\
    \qquad\land
    \cn{person}^p(s_1,a_0)\land
    (s_0,\cnc{c}_1)\prec(s_1,\cnc{c}_1)\land
    \cn{uniqAtS}(k_0, s_1,\cnc{c}_1)
  \end{array}\]
  \caption{DoorSEP First {\sas}}\label{fig:doorsep first sas}
\end{figure}

The next step in the analysis makes use of the trust axiom.  The
result of the {\cpsa} analysis is transformed into the {\sas} shown
in Figure~\ref{fig:doorsep first sas}.  The antecedent specifies the
initial scenario described by the first model.  The consequence
specifies what else must be added to make the initial scenario into
the complete execution show in Figure~\ref{fig:doorsep first shape}.

When the {\sas} is added to the current theory, Razor finds one
model.  The skeleton extracted from this model is very similar to the
shape in Figure~\ref{fig:doorsep first shape} with one crucial
difference: the key $D'$ is uncompromised.  This skeleton is
unrealized, so {\cpsa} can make a contribution.  It finds a {\sas}
that extends the length of the person strand to full length and
equates~$D$ and~$D'$.  The addition of this {\sas} produces a model
that characterizes a realized skeleton with full agreement between the
door and person strands.  Because the skeleton is realized, {\cpsa}
has nothing more to contribute and the analysis terminates.

%% file: model-finding.tex
\section{Finding Minimal Models}\label{sec:finding minimal models}

Below we will assume familiarity with basic ideas and results from
first-order mathematical logic; notions that are not defined here are
treated in any text on logic.  There are some variations in the
formalities of different presentations in the literature, but
definitions in the SMT-Lib standard \cite{smt-lib-web} line up
particularly well with our notation.

In this chapter we present some of the foundations of model-finding,
focusing on the use of a Satisfiability Modulo Theories (SMT) solver.
In broadest terms, model-finding is the following task:   given a
logical theory \thyT, produce one or more (finite) models of $\thyT$.

Of course a typical satisfiable theory will have many models.  Special
emphasis is given in this chapter to the question of \emph{which
  models should be presented to the user?}  One answer---embodied in
the \lpa\ tool---is based on the fundamental notion of
\emph{homomorphism} between models, with a special emphasis on models
that are \emph{minimal} (see Section~\ref{sec:minimality}) in the
pre-order determined by homomorphism.

Section~\ref{sec:working-with-solver} addresses some of the strategies
we evolved in programming against an SMT solver.


\subsection{Foundations}
\label{sec:foundations}

\subsubsection{Language}

We work with many-sorted first-order logic.
Concretely, we implement the logic defined in the SMT-Lib standard,
version 2, as described in \cite{smtlib}.   This logic is \emph{not}
order-sorted: sorts are all disjoint, and models must have all
sort-interpretations be non-empty.

The SMT-LIB language adopts the convention that all expressions are
terms, and the ``formulas'' are, by convention, terms of sort Bool.
For software-engineering reasons we have found it beneficial to work
with terms and formulas as separate types.
The translation is straightforward, of course, and so Razor does the
following
\begin{itemize}
\item when theories are loaded from user's input files (or
  intermediate files created by our tool), the theory is translated
  from a ``term-based'' one to a ``term-and-formula--based'' one, which
  is manipulated internally;
\item when theories are sent to the SMT solver
  they are translated back to being
  a ``term-based'' one
\end{itemize}

Certain classes of formulas and theories play a special role in our tool.

\begin{definition}[PE formula, Geometric theory]
A formula is \emph{positive-existential,} or \emph{PE}, if it is built
from atomic formulas (including $\truth$  and $\falsehood$) using
$\land$, $\lor$ and $\exists$.

A theory $T$ is \emph{geometric} if has an axiomatization in which
each axiom is of the form
\[
  \forall \vec{x} .\quad \alpha(\vec{x}) \to \beta(\vec{x})
\]
where $\alpha$ and $\beta$ are positive-existential.
\end{definition}

\subsubsection{Models}

Fix a signature \sig.  A model \mM for signature \sig will be called
a \emph{\sig-model}; notation $\mM \models \sig$.

Let $T$ be an \sig-theory.
We are interested in categories of models of $T$; for the most part we
only consider finite models.

\begin{notation}
  If $\alpha$ is a formula with free variables among %
$x_1, \dots, x_n$ %
and $a_1, \dots , a_n$ are elements of a model
  \mA we will sometimes write %
  $\mA \models \alpha[a_1, \dots , a_n]$ as shorthand  to mean that
  $\mA \models \alpha(x_1, \dots , x_n)$ under the environment sending
  each $x_i$ to $a_i$.
\end{notation}

\begin{definition}[Diagram, Characteristic Sentence]%

  Let $\mM$ be a finite model over signature \sig. %
  Expand the signature \sig by adding, %
  for each element $e$ of the domain of \mM, a constant $c_e$, and let
  $\mM^+$ be the corresponding expansion of $\mM$.  The \emph{diagram}
  $\diagram{\mM}$ %
  of \mM is the set of atomic sentences and negations of atomic
  sentences true in $\mM^{+}$.  %
  If we take only the atomic sentences, the result is the
  \emph{positive diagram} $\posdiagram{\mM}$ of \mM.

    The sentence---in the original signature \sig---obtained by
    converting the new constants in the positive diagram to variables and
    existentially quantifying them is called the %
    \emph{characteristic sentence} $\ch{\mM}$ of  \mM.

    If we enrich $\ch{\mM}$ with a suitable conjunct  expressing the
    fact that the elements named by the constants are all distinct, we
    obtain the    \emph{$i$-characteristic sentence} $\ich{\mM}$ of
    \mM .
\end{definition}

\subsubsection{Homomorphisms}

\begin{definition}
  Let \mA and \mB be \sig-models.  A function %
  $h : |\mA| \to |\mB|$ is a \emph{homomorphism} if
  \begin{enumerate}
  \item %
    $\mA \models f [a_1, \dots, a_n] = a$ \text{ implies }
    $\mB \models f [h(a_1),  \dots, h(a_n)] = h(a)$ and
  \item
    $\mA \models R [a_1, \dots, a_n]$ \text{ implies }
    $\mB \models R [h(a_1),  \dots, h(a_n)]$.
  \end{enumerate}
  It is a \emph{strong homomorphism} if the second condition is
  replaced by%
  \footnote{Caution: some authors use the term ``homomorphism'' to
    mean our ``strong
    homomorphism.''}

  \begin{enumerate}
  \item[2']
    $\mA \models R [a_1, \dots, a_n]$ \text{ if and only if }
    $\mB \models R [h(a_1),  \dots, h(a_n)].$
  \end{enumerate}
  ( The first condition is already equivalent to ``if and
  only if.'')
\end{definition}

\begin{definition}
  A homomorphism from a model to itself is called an \emph{endomorphism.}

  An \emph{embedding} $e : \mA \to \mB$ is a strong homomorphism which
  is injective.


  An \emph{isomorphism} $f: \mA \to \mB$ is a homomorphism such that
  there is a homomorphism $g: \mB \to \mA$ with $g \comp f = id_\mA$
  and $f \comp g = id_\mB$.
\end{definition}

Note that an isomorphism is \emph{not} the same as a bijective
homomorphism (it is the same as a bijective \emph{strong} homomorphism).
On the other hand:

\begin{lemma} \label{bi-inj-iso}
  Let \mA and \mB be finite.
  If $h: \mA \to \mB$ and $g: \mB \to \mA$ are injective
  homomorphisms, then they are isomorphisms.
\end{lemma}
\begin{proof}
  It will be enough to prove that $h$ is an isomorphism. %
  We use the following elementary fact: if $f : X \to X$ is an
  injective set-theoretic function on a finite set, then for some $n$,
  $f^n : X \to X$ is the identity.

  So consider $g\comp h: \mA \to \mA$.  This is an injection, so for
  some $n$, $(g \comp h)^n$ is the identity on \mA.  So
  $(g \comp h)^{n-1} \comp g$ is $h^{-1}$.  Since this function is a
  composition of homomorphisms, it is a homomorphism.  Thus $h$ is an
  isomorphism.
\end{proof}

\subsubsection{Categories of Models}
Since the identity map is injective and injective maps are closed
under composition, the class of \sig-models under injective maps makes a
category.

\begin{notation} \hfill
\begin{itemize}
\item $\modcat{\sig}$ is the category of all finite models of \sig
  with arbitrary homomorphisms.

Write $\mA \homleq \mB$ %
(sometimes $\mB \homgeq \mA$) %
  if there is a $\modcat{\sig}$ map $h:
\mA \to \mB$.

Write $\mA \homeq \mB$ if
$\mA \homleq \mB$  and $\mB \homleq \mA$.

Write $\mA \homle \mB$ if
$\mA \homleq \mB$  and not $\mB \homleq \mA$.

\item $\imodcat{\sig}$ is the category of all finite models of \sig
  with injective homomorphisms.

Write $\mA \ihomleq \mB$ %
(sometimes $\mB \ihomgeq \mA$) %
if there is a $\imodcat{\sig}$ map $h:
\mA \to \mB$.

 Write $\mA \ihomeq \mB$ if
$\mA \ihomleq \mB$  and $\mB \ihomleq \mA$.

Write $\mA \ihomle \mB$ if
$\mA \ihomleq \mB$  and not $\mB \ihomleq \mA$.
\end{itemize}
\end{notation}
Both \homleq and \ihomleq are preorders.

Note that if $\mA \ihomleq \mB$ then $\mA$ is isomorphic to a (not
necessarily induced) submodel of $\mB$.

\begin{definition}
  Model \mB is a \emph{submodel} of \mA if
  $|\mB|  \subseteq |\mA| $ and the inclusion function is a
  homomorphism.

  A submodel \mB is a \emph{proper submodel} of \mA if
  $\mB \neq \mA$:  this can happen if \emph{either}
  $|\mB| \neq |\mA| $ \emph{or} for some
  tuple $[a_1, \dots, a_n]$,
    $R^{\mA}[a_1, \dots, a_n]$ fails while
    $R^{\mB}[(a_1), \dots, (a_n)]$ holds.

  Model \mB is an \emph{induced submodel} of \mA if
  $|\mB|  \subseteq |\mA| $ and the inclusion function is a
  strong homomorphism.

\end{definition}

\begin{definition}
  If \modclass is a class of \sig-models and
  $\modclass_0 \subseteq \modclass$ say that $\modclass_0$ is a
  \emph{$\modcatsig$ set of support for \modclass} if %
  for all $ \mB \in \modclass$, there exists $\mA \in \modclass_0$
  with $\mA \homleq \mB$. %
  Similarly for $\imodcatsig$.
\end{definition}

\subsubsection{Well-foundedness of the homomorphism preorder}
\label{sec:well-foundedness-hom}
We will often add axioms to a theory to ensure that there is an upper
bound on the size of its models.  In such a case there will be only
finitely many models of $T$.

\begin{lemma}
  Let $T$ be a theory with only  finitely many models.
  Then the \homle and \ihomle orders on models of $T$ are well-founded.
\end{lemma}
\begin{proof}
  Suppose for the sake of contradiction that we have an infinite
  descending chain of strict homomorphisms:
\begin{align*}
\dots & \homle \mM_2 \homle \mM_1 \homle \mM_0 %
\end{align*}
Then we have $\mM_{i+k} \homle \mM_i$ for any $k \geq 0$.  Since $T$
has finitely many models, we eventually get $i$ and $k \geq 0$ with
$\mM_{i+k+1}$ isomorphic to $\mM_i$.  So
$\mM_{i+k+ 1} \homle \mM_{i+1}$.  But that implies
$\mM_{i} \homle \mM_{i+1}$, a contradiction.

The same argument applies to \ihomleq as well.
\end{proof}

\subsubsection{Homomorphisms and Logical Form}

\begin{theorem}
  The following are equivalent, for a formula $\alpha(\vec{x})$:
  \begin{enumerate}
  \item $\alpha$ is preserved by homomorphism: if $h : \mA \to \mB$ is
    a homomorphism, and $\vec{a}$ is a vector of elements from \mA
    such that $\mA \models \alpha[\vec{a}]$, then
    $\mB \models \alpha[\vec{h a}]$.
  \item $\alpha$ is logically equivalent to a PE formula.
  \item $\alpha$ is equivalent, in the category $\modcat{\sig}$ of
      finite models, to a PE formula.
  \end{enumerate}
\end{theorem}
\begin{proof}
  The equivalence of (1) and (2) is a classical result in model theory
  when considering arbitrary models.  The equivalence of (1) and (3)
  is a deep result of Rossman \cite{Rossman2008}.
\end{proof}

\begin{lemma}
  Let \mM and \mN be  arbitrary models.

  \begin{itemize}
  \item
  The following are equivalent.
  \begin{enumerate}
  \item   $\mM \homleq \mN.$
  \item   $\mN \models \ch{\mM}.$
  \item  $\ch{N} \models \ch{\mM}$.
  \end{enumerate}
\item   The following are equivalent.
  \begin{enumerate}
  \item   $\mM \ihomleq \mN.$
  \item     $\mN \models \ich{\mM}.$
  \item  $\ich{N} \models \ich{\mM}$.
  \end{enumerate}
\end{itemize}
\end{lemma}
\begin{proof}
  For the assertions about unconstrained homomorphisms:

  (1) implies (2):
  Any function $h: |\mM| \to |\mN|$ determines a way to instantiate in \mN
  the existentially quantified variables in \ch{\mM}.   When $h$ is a
  homomorphism, this instantiation will make the body of \ch{\mM}
  true, essentially by definition of homomorphism.

  (2) implies (1):  Suppose $\mN \models \ch{\mM}$.  Then the
  instantiation of the existentially quantified variables in \ch{\mM}
  determines a function $h: |\mM| \to |\mN|$.  The fact that this
  instantiation makes the body of \ch{\mM} true implies that $h$ is a
  homomorphism.

  (1) implies (3): Suppose $\mM \homleq \mN$.  We want to show that
  for any \mP, $\mP \models \ch{\mN}$ entails $\mP \models
  \ch{\mM}$.  %
  By the equivalence of (1) and (2) it suffices to show that
  $\mN \homleq \mP$ entails $\mM \homleq \mP$.   Since
 $\mM \homleq \mN$ this follows from transitivity of $\homleq.$

  (3) implies (2): this is clear, since \mN is a model of \ch{\mN}.

  The proofs for the injective case are almost identical.  The
  additional factor is that the ``distinctness'' assertions in
  the i-characteristic sentences ensure that the homomorphisms
  considered are injective.
\end{proof}

\subsubsection{Retractions and Cores}

The notion of the \emph{core} of a model is standard; it is important
for us because
cores give canonical representatives of $\homeq$ equivalence classes.

Core are defined in terms of \emph{retractions,} as follows.

\begin{definition}
  A \emph{retraction} $r: \mA \to \mB$ is a homomorphism such that
  there is a homomorphism $e: \mB \to \mA$ with
  $r \circ e = \id{B}$.
\end{definition}

The following are well-known facts about retractions.
If $r$ is a retraction then $r$ is surjective and the corresponding $e$ is
 injective.  Indeed $e$ is an embedding (though $r$ need not be a
strong homomorphism).

If $r$ is an endomorphism then $r$
is a retraction if and only if $r$ is idempotent: $r \comp r = r$;
this is to say that $r$ is the identity on its image.

It is false in general that if $\mA$ is a model of $T$ then  a
retraction of
  \mA is a model of $T$.  But it does hold if $T$ is geometric.

\begin{lemma} \label{geo-retract}
Let $T$ be a geometric theory, %
$\mA \models T$ , %
and $r : \mA \to \mB$ a retraction.
Then $\mB \models T$.
\end{lemma}
\begin{proof}
  Let $e : \mA \to \mB$ satisfy $r \comp e = id_\mB$.  %
  Consider an axiom $\sigma$ of $T$   true in \mA
  \[
    \sigma \equiv \; \forall \vec{x} .\; \alpha(\vec{x}) \to \beta(\vec{x})
  \]
 where $\alpha$ and
 $\beta$ are positive-existential formulas.  %
  To show $\sigma$ is true in \mB, consider a tuple $\vec{b}$ of
  elements such that $\alpha[\vec{b}]$ is true in \mB.
  Since PE formulas are preserved by homomorphisms,
  $\mA \models \alpha[\vec {e b}]$.   Since $\mA \models \sigma$,
  $\mA \models \beta[\vec {e b}]$.
  Since PE formulas are preserved by homomorphisms,
  $\mB \models \beta[\vec { r e b}]$.  %
  Since  $r \comp e = id_\mB$,
  $\mB \models \beta[\vec {b}]$, as desired.  %
\end{proof}

The following will be useful later.
\begin{lemma} \label{power-endo-retract}
 If $\mA$ is finite and $h: \mA \to \mA$ is an endomorphism then some power of $h$ is a
 retraction.
\end{lemma}
\begin{proof}
  By induction on the size of \mA.
  If $h$ is injective, then some power of $h$ is the identity, a
  retraction.

  Otherwise suppose that $h(e) = h(e')$ for some $e, e' \in |\mA|$.
  Let $\mA'$ be the image of $h$, and let $h'$ be $h$
  restricted to $|\mA'|$.

  Then $|\mA'| < |\mA|$ and $h'$ is an endomorphism of $\mA'$.
  By induction, for some $n$, $(h')^n$ is a retraction.
  It follows that $h^{2n}$ is a retraction; to show that it is
  idempotent we calculate:
  \begin{align*}
    h^{2n} \comp h^{n2} &= h^{3n} \comp h^{n} \\
           &= (h')^{3n} \comp h^{n} &\text{since $h \comp h = h' \comp
                                      h$} \\
           &= (h')^{n} \comp h^{n} &\text{since $ (h')^{2n} =
                                     (h')^{n}   $}  \\
           &= h^{2n} &\text{since $h' \comp h = h \comp
                                      h$}
  \end{align*}

\end{proof}

\begin{corollary} 
  If $h: \mA \to \mA$ is a non-injective endomorphism then some power
  of $h$ is a proper retraction.
\end{corollary}
The corollary was observed by Gottlob in \cite{gottlob2005}.
It is useful in computing cores, as we will see later.

\begin{definition}
A submodel  \mC of \mA is a \emph{core} of \mA if there is a
retraction  $r: \mA \to \mC$ but no retract
$r': \mA \to \mC'$ for any proper submodel $\mC'$ of $\mC$.

   A model \mC is a \emph{core} if it is a core of itself.
\end{definition}

By \reflem{geo-retract}, if $T$ is a geometric theory, and $\mA$ is a
model of $T$ then the core of \mA is a model of $T$.

The following are well-known (see, for example, \cite{hell1992} in the case of graphs). %
\begin{lemma}  \label{core-tfae}
  Let \mA be finite.
  \begin{enumerate}
  \item \mA has a core.
  \item If \mC is a core of \mA then $\mA \homeq \mC$.
  \item A core of \mA is an induced submodel of \mA.

  \item \label{part:no-retracts} %
    \mC is a core if and only if it has no proper retracts.
  \item \label{part:no-endo}
    \mC is a core if and only if
    it has no proper endomorphisms (equivalently,
    every endomorphsim of \mC is an
    embedding;  equivalently
    every endomorphsim of \mC is an automorphism).
  \item
  A  submodel \mC of \mA is a core of \mA if
  there is an endomorphism $h: \mA \to \mC$ but no endomorphism
  $\mA \to \mC'$ for any proper submodel $\mC'$ of $\mC$.

  \item If \mC and \mC' are cores of a model \mA then \mC and \mC' are
    isomorphic.

  \end{enumerate}
\end{lemma}
\begin{proof} \hfill
  \begin{enumerate}
  \item This holds for simple cardinality reasons.
  \item This holds by definition of retraction.
  \item This holds since the retraction is the identity on the core.
  \item This is immediate from the definition.
  \item By \reflem{power-endo-retract}, a model \mC has no proper retracts if
    and only if it has no proper endomorphisms.  Now apply the
    previous part.
  \item If $\mC$ is a core of \mA then by definition there is an
    endomorphism $h: \mA \to \mC$; if there were
    $h': \mA \to \mC'$ to a  proper submodel $\mC'$ of $\mC$ then by
    \reflem{power-endo-retract}, there would be a retraction to \mC',
    contradicting the definition of core.

    If  $h: \mA \to \mC$ but there is no endomorphism
    $\mA \to \mC'$ for any proper submodel $\mC'$ of $\mC$, then
    \mC itself has no proper endomorphisms, so by part
    \ref{part:no-endo},  \mC is a core.

  \item Suppose $r: \mA \to \mC$ and   $r': \mA \to \mC'$ are
    retracts.
    Consider the map $r\restriction_{\mC'}$, $r$ restricted to \mC'.
    This is a homomorphism from \mC' to \mC.
    Similarly %
    $r'\restriction_{\mC} : \mC \to \mC'$.
    The composition %
    $(r\restriction_{\mC'}) \comp (r'\restriction_{\mC} )$
    is an endomorphism of \mC, hence is injective. %
    The composition in the other order is injective as well,
    so, as observed in \reflem{bi-inj-iso}, each is an isomorphism.
  \end{enumerate}
\end{proof}

\subsubsection{Computing Cores}
\label{sec:computing-cores}

Testing whether a model is a core is NP-complete.
So computing cores is apparently difficult, from a worst-case complexity
perspective.
But it is not difficult, using an SMT solver, to write a program that
behaves well in practice.

\begin{definition}
  If \mM is a finite model for signature $\Sigma$, the sentence %
  $\ehom{\mM}$, over the signature $\Sigma_h$ that extends $\Sigma$ by
  adding a new function symbol
  $h_s: S \to S$ at each sort $S$, is the conjunction of
  \begin{itemize}
  \item the diagram of \mM,
  \item the sentence expressing ``$h$ is a homomorphism'', and
  \item the sentence expressing ``$h$ is not injective.''
  \end{itemize}
\end{definition}
\begin{algo}[ComputeCore] \label{alg:core} \hfill
\begin{itemize}
\item[] \textbf {input:}  model $\mM$ over signature $\Sigma$
\item[] \textbf{output:}  a core \mP of \mM
\item[] \textbf{initialize:}
  Set \mP to be \mM
\item[] \textbf{while} $\ehom{\mP}$ is satisfiable
  \begin{itemize}
  \item[]
    let $\mP'$ be a model of $\ehom{\mP}$; \\
    let $\mP_0$ be the image of $\ehom{\mP}$ in \mP'; \\
    let $\mP$ be the reduct of $\mP_0$ to the original signature $\Sigma$
  \end{itemize}
  \item [] \textbf{return} \mP
\end{itemize}
\end{algo}
\begin{lemma}
 \refalg{alg:core} computes a core of its input.
\end{lemma}
\begin{proof}
  The algorithm terminates because the size of the model \mP decreases
  at each iteration.  The resulting model is a core by
  part~\ref{part:no-endo} of \reflem{core-tfae}.
\end{proof}

\subsection{Minimality}
\label{sec:minimality}

\begin{definition}
  Let \modclass be a class of models closed under homomorphisms.

  A model \mA is \imin for $\modclass$ if it is a minimal element in the \ihomleq
  preorder on models in \modclass.    %
  Equivalently, whenever $\mB \ihomleq \mA$ then
  $\mB \ihomeq \mA$.

  A model \mA is \amin for $\modclass$ if it is a minimal element in the \homleq
  preorder on models in \modclass.    %
  Equivalently, whenever $\mB \homleq \mA$ then
  $\mB \homeq \mA$.

\end{definition}

Typically we are interested in the case when \modclass is the class of
models of a theory $T$; in this case we may use the phrases
$T$-minimal, or $T$-i-minimal.

Here are some local
characterizations  of $i$-minimality.
\begin{lemma} \label{imin-iso}
  The following are equivalent for a model \mA in \modclass.
  \begin{enumerate}
  \item \label{one} %
    $\mA$ is \imin for $\modclass$.
  \item \label{two} %
    No proper submodel of
    \mA is in \modclass.
  \item \label{three} %
    Whenever \mB is in \modclass and $k: \mB \to \mA$ is an injective
    homomorphism, then $h$ is an isomorphism.
  \end{enumerate}
\end{lemma}
\begin{proof}

  For \ref{one} $\Rightarrow$ \ref{two}, if \mB were a proper submodel
  of \mA in \modclass then the inclusion map would contradict
  $i$-minimality of \mA.

  For \ref{two} $\Rightarrow$ \ref{three}, suppose
  $k: \mB \to \mA$ is
  injective, with $\mB \in \modclass$.  The image of $k$ in \mA is
  isomorphic with \mB, hence is in \modclass, so by \ref{two} is
  not a proper submodel.

  The implication \ref{three} $\Rightarrow$ \ref{one} is easy.
\end{proof}

\subsubsection{Submodel-minimality}

One could imagine yet another notion of minimality, where the preorder
on models is given by the submodel relation.  This notion could be described
without any reference to homomorphisms (the notion of submodel could
be defined natively, although we have not done so):
for a theory $T$,
a model \mA of $T$ is ``submodel-minimal'' for $T$ precisely if no
proper submodel of \mA is a model of $T$.

But \reflem{imin-iso} says that this notion is precisely the same as
$i$-minimality.

This observation will be useful when we turn to \emph{computing}
$i$-minimal models.

\subsubsection{Relationships between \amin and \imin}
\label{sec:relat-betw-amin}

\paragraph{An  \imin model is not necessarily \amin.}
\begin{example}
Let $T$ be
\[ \exists x . P(x) \; \land \; \exists x . Q(x)
\]

Let $\mA$ have one element $a$ with
\[\mA \models P[a] \land Q[a]
\]

Then \mA is \imin but not \amin.
The model \mB with two elements
$a_1$ and $a_2$ such that
\[
\mB \models A[a_1] \land B[a_2]
\]
is strictly below $\mA$ in the \homleq preorder.
(\mB is $a$-minimal for $T$.)
\end{example}

\paragraph{An  \amin model is not necessarily \imin.}
\begin{example}
Let $T$ be
\[ \exists x . P(x)
\]
Let $\mA$ have two elements $a_1$ and $a_2$ with
\[\mA \models P[a_1]
\text{ and }
\mA \models P[a_2]
\]
Then \mA is \amin. %
But \mA is not \imin:
the induced model determined by $a_1$
is a model of $T$.
\end{example}

\paragraph{However, an \amin model which is a core \emph{will} be \imin.}

\begin{lemma} \label{amin+core-is-imin}
  If $\mA$ is \amin for $T$ and is a core, then
\mA is \imin for $T$.
\end{lemma}

\begin{proof}
  Suppose \mB is a model of $T$ and %
  $j: \mB \to \mA$ is injective.
  Since \mA is \amin, there is a homomorphism
  $h : \mA \to \mB$.
  The composition $j \comp h$ is an endomorphism of \mA.
  Since \mA is a core this map is injective, so
  $h$ is injective, and $\mA \ihomeq \mB$.
\end{proof}

Lemma~\ref{amin+core-is-imin} is attractive in the sense that it
suggests striving for the best of both worlds ($i$-minimality and
$a$-minimality simultaneously) by computing $a$-minimal cores.
Unfortunately, if a theory $T$ fails to be geometric then the core of
a model of $T$ can fail to be a model of $T$.  Thus we cannot in
general construct cores in our model-finding.

\begin{example}
  Let $T$ be the theory that says $\exists x. P(x)$ and that there
  exist exactly two elements (the latter sentence is not geometric).
  The model \mM with two elements, with $P$ holding of each of them,
  is $a$-minimal, but not $i$ minimal.  The model with two elements
  and $P$ holding of just one of them is $i$-minimal.   The core of
  $\mM$ has one element, with $P$ holding, but this is not a model of $T$.
\end{example}

\paragraph{$i$-minimizing an $a$-minimal model}

Suppose \mM is $a$-minimal.  Let \mK be an $i$-minimal model below
\mM, that is, \mK is $i$-minimal and there is an injective
$h : \mK \to \mM$.  Then \mK is both $i$-minimal and
$a$-minimal.  It is $a$-minimal because it is $a$-hom equivalent to
\mM (we have $h : \mK \to \mM$ by assumption, and there is a map
$\mM \to \mK$ by virtue of \mM being $a$-minimal).  So the true
best of both worlds is: first $a$-minimize, then $i$-minimize.









\subsubsection{Computing  $i$-minimal Models}

As noted in \reflem{imin-iso}, finding  an \imin model for $T$ which
is
$\ihomleq$ given a model \mM is equivalent to finding
a minimal \emph{submodel} of \mM satisfying $T$.

This leads to the following procedure, originally developed for use in
the \emph{Aluminum} tool \cite{nelson_ICSE13}
\begin{itemize}
\item fix the universe
\item fix the negative literals
\item negate some positive literals
\end{itemize}
until no change.

Recall that $\domain{C}$ is a sentence expressing the fact that
every element of the domain(s) of a model is named by a constant in $C$.

For this algorithm we use the notation $\flip{\mP}$ to denote
\begin{align*}
& \bigwedge \dset{\neg \alpha} {\alpha \text{ is an atomic sentence, }
  \mP \models \neg \alpha}
\\ \land
& \bigvee \dset{\neg \beta} {\beta \text{ is an atomic sentence, }
  \mP \models \beta}
\end{align*}
Note in particular that if $c$ and $c'$ are constants naming distinct
elements of a model \mP, then $c \neq c'$ is one of the conjuncts of
$\flip{\mP}$.

\begin{algo}[i-Minimize] \label{alg:imin} \hfill
  \begin{itemize}
  \item[] \textbf{input:} theory $T$ and model $\mM \models T$
  \item[] \textbf{output:} model $\mP \models T$ such that %
    $\mN$ is \imin for $T$ and $\mP \ihomleq \mM$
  \item[] \textbf{initialize:} set \mP to be \mM
  \item[] \textbf{while}
    $T' \eqdef T \cup \{ \flip{\mP} \} $ is satisfiable
    \begin{itemize}
    \item[] set \mP to be a model of $T'$
    \end{itemize}
    \item[] \textbf{return} \mP

  \end{itemize}

\end{algo}
\begin{lemma}
  \refalg{alg:imin} is correct:  if \mM is a finite model of $T$ then %
      \refalg{alg:imin} terminates on \mM, and %
      the output \mP is an $i$-minimal model of $T$ with
      $\mP \ihomleq \mM$
\end{lemma}
\begin{proof}
   Each iteration goes down in the $\ihomle$ ordering,
  thus termination.   To show that the result is \imin for $T$, it
  suffices, by \reflem{imin-iso}, to argue that the result is a
  minimal $T$-submodel of the input, under the submodel ordering.   But
  this is clear from the definition of the sentences $\flip{}$.
\end{proof}

\subsubsection{Computing  $a$-minimal Models}

What about computing a-minimal models?  That's harder.
First of all, for a given theory there might be no finite a-minimal
models at all.  Example: the theory with one unary function and no
axioms.   The initial (hence unique minimal) model of this theory is
the natural numbers.
Another way to put this is:   the $\homle$ preorder is not well-founded in
general.

If we bound the size of the domain(s) of our models then a-minimal
models exist: the $\homle$ preorder is well-founded, so the set of
minimal elements with respect to this order is non-empty.
The question is, how to compute \amin models?

 We show how to do two things
\begin{enumerate}
\item Given model $\mA$ define a  sentence
  $\homTo{\mA}$ such that %
  $\mB \models \homTo{\mA}$ if and only if $\mB \homleq \mA$.


\item Given model $\mA$ define a sentence
$\homFrom{\mA}$ such that
$\mB \models \homFrom{\mA}$ if and only if $\mA \homleq \mB$.

We define
 $\avoid{\mA}$ to be $\neg \homFrom{\mA}$, so that
$\mB \models \avoid{\mA}$ if and only if
  $\mA \not \homleq \mB$.


\end{enumerate}
Then the process of finding an $a$-minimal model for $T$ below a given model
$\mA$ is to iterate the process of constructing a model that is
strictly below \mA in the \homleq ordering, which is to say, a model
\mB such that there is a homomorphism from \mB to \mA but not
homomorphism from $\mA \to \mB$.

That is:
\begin{quote}
\em
  given $\mA \models T$, construct the theory %
  \[
    T \cup \{ \homTo{\mA} \} \cup \{ \avoid{\mA} \}
\]
If this is unsatisfiable then \mA is $T$ minimal. %
If this is satisfiable, let \mA' be a model of this theory, and iterate.
\end{quote}

Important: this process is not guaranteed to terminate for an
arbitrary $T$.  But as observed earlier,
if at the outset we bound the size of the models to be considered,
the process will terminate.

The challenge---suggested earlier---is to define hom-to and hom-from
sentences with as much ``existential'' character as we can manage.

\subsubsection{Hom To}

This is straightforward ``solver programming''.

Given model \mM.  We want to characterize those \mB such that there is a
hom $h : \mB \to \mM$, by constructing a sentence
$\homTo{\mM}$.




\begin{algo}[HomTo] \hfill
\begin{itemize}
\item[] \textbf {input:} model $\mM$ over signature $\Sigma$.
\item[] \textbf{output:} sentence $\homTo{\mM}$ in an expanded signature $\Sigma^{+}$,
  such that for any model $\mP \models \Sigma$,
  $\mP \homleq \mM$ iff there is an expansion $\mP^+$ of $\mP$ to $\Sigma^+$
  with $\mP^+ \models \homTo{mM}$.
\end{itemize}
\begin{enumerate}
\item[] \textbf{define}  $\Sigma^+$ to be the extension of $\Sigma$
  obtained by
  \begin{itemize}
  \item[] adding a set of fresh constants in one-to-one correspondence
    with the elements of the domain of $\mM$
  \item[] adding a function symbol $h_S : S \to S$ at each sort $S$
  \end{itemize}
\item[] \textbf{define}  $\homTo{\mM}$ as the conjunction of the
  following sentences, one for each
  function symbol $f$ and predicate $R$ in $\Sigma$.
  \begin{align}
    &\forall \vec{x} , y .\; f \vec{x} = y
      \implies
      \bigvee \{ (\vec{hx} = \vec{e} \land y = e') \mid
      \mM \models f \vec{e} = e'  \}
    \\
    &\forall \vec{x} .\; R \vec{x} = true  \implies
      \bigvee \{ (\vec{hx} = \vec{e} ) \mid
      \mM \models R \vec{e} = true \}
  \end{align}
\end{enumerate}
\end{algo}
\begin{lemma}
  Suppose $\mM$ and $\mB$ are $\sig$ models.
  There is a $\sig$ hom $h: \mB \to \mM$ iff
  there is a model $\mB^{+} \models \homTo{\mM}$ such that
  \mB is the reduction to \sig of $\mB^{+}$.
\end{lemma}

\begin{proof}
  Suppose \mB is the reduction of  $\mB^{+} \models \homTo{\mM}$.
  The interpretation of $h$ in $\mB^{+}$ defines a function from $|\mB|$
 to $|\mM|$.   %
 We want to show $h$ is actually a $\Sigma_u$ hom. %
 But that's just what $\homTo{\mM}$ does.

Suppose $\mB \models \sig$ and there is a hom %
$h : \mB \to \mM$.  %
We want to show that there is an expansion $\mB^+$ of \mB
satisfying
 $\mB \models \homTo{\mM}$. %
The actual homomorphism $h$ determines the interpretation  in
$\mB^{+}$ of the symbol
$h$
and the interpretation of the new constants $c'$.
And since $h$ is a homomorphism, the clauses in $\homTo{\mM}$ are
satisfied.
\end{proof}


\subsubsection{Hom From and Avoid Cone}
\label{hom-from-and-to}
Our goal is: given a model $\mM$, find a formula to capture
\textbf{not} being in the hom-cone of  $\mM$.

\emph{This is more interesting that the hom-to problem, because we are
  going to \textbf{negate} the sentence we build, to express
  hom-cone-avoidance.  So we want to minimize the number of
  existential quantifiers we use here.}

The ideal outcome would be to construct an existential sentence capturing
the complement of the hom cone of \mM.   %
 Equivalently we might look for a structure $\mB$ such that for any
$\mA$, $\mM \homleq \mA$ iff $\mA \not \homleq \mB$.   This is called ``homomorphism
duality'' in the literature.
Such a structure  doesn't
always exist; and even when it does, it can be exponentially large in
the size of \mM \cite{ErdosPTT17}.  So we turn to heuristic methods.

The strategy is to construct a sentence guaranteed to capture the
hom-from problem, then refine this sentence to eliminate (some) quantifiers.

Start with the $C$-rules of the standard model rep for \mM
\begin{align*}
  \text{equations} & \quad f \vec{c} \to c
\end{align*}
where the $c$ are over $\Sigma_{uK}$
but there is exactly once $c$ per element.

Convert that  to the \sig sentence
\[
  \rep{\mM} \;  \eqdef \;
  \exists \vec{x} , \;f \vec{u}  = u
\]
by replacing the $c$  by variables.

\begin{lemma} \label{modelrep-correct}
  Let $\mM$ and $\mF$ be $\sig$ models.
  Then $ \mM \homleq \mF$ iff
  $\mF \models \rep{\mM}$.
\end{lemma}
\begin{proof}
  If there is a hom $h: \mM \to \mF$: use the fact that homomorphisms preserve
  positive existential formulas.  (The image of $h$ on the $\vec{x}$
  says how to interpret the $\vec{x}$ in $\mF$.)

 If
 $\mF \models \rep{\mM}$ :
 since there is one variable per
 element of $\mM$, the interpretation of the
 $\vec{x}$ says how to define
 a function $h : |\mM| \to |\mF|$.
 The fact that the equations comprising the body  of
 $\rep{\mM}$ completely describe the graphs in \mM of the non-Boolean
 functions of \sig and the tuples making the predicates true in \mM
 ensure that $h$ makes a homomorphism.
\end{proof}

That's fine, but there are as many existential quantifiers in
$\rep{\mM}$ as there are domain elements.  If we were to take $\homFrom{\mM}$ to
be $\rep{\mM}$, and  define $\avoid{\mM}$ by
simply negating this would lead to a sentence inconvenient for the SMT
solver.  We can compress the representation, though.  This will lead
to a nicer representation sentence, which we will take as $\homFrom{mM}$.


\begin{algo}[HomFrom]\label{alg:hom from} \hfill
  \begin{itemize}
  \item[] \textbf {input:}  model $\mM$ over signature $\Sigma$
  \item[] \textbf{output:} sentence $\homTo{\mM}$ over signature $\Sigma$, such that
    for any model $\mP \models \Sigma$, $\mM \homleq \mP$ iff $\mP \models \homFrom{\mM}$.
  \item[] \textbf{comment:} sentence $\homFrom{\mM}$ is designed to use as few
    existential quantifiers as possible,
    in a ``best-effort'' sense.
  \item [] \textbf{initialize:}
    Set  sentence $\homFrom{\mM}$ to be $\rep{\mM}$, %
    the standard model representation sentence
    for \mM.
  \item[] \textbf{while}
    there is a conjunct in the body of $\homFrom{\mM}$ of the form
    \[
      f(t_1, \dots, t_n) = x
    \]
    such that  $x$ does not occur in any of the $t_i$,
    \begin{itemize}
    \item[]
      replace all
      occurrences of $x$ in $\homFrom{\mM}$ by $f(t_1, \dots, t_n)$.
      Erase the resulting trivial equation
      $f(t_1, \dots, t_n) = f(t_1, \dots, t_n) $
      and erase the $(\exists x)$ quantifier in front.
    \end{itemize}
  \end{itemize}
\end{algo}
\newpage
\begin{lemma}
  For any model $\mP \models \Sigma$, $\mM \homleq \mP$ iff $\mP \models \homFrom{\mM}$.
\end{lemma}
\begin{proof}
  By \reflem{modelrep-correct} the assertion is true at the
  initialization step.  So it suffices to observe that each
  transformation of $\homFrom{\mM}$ yields a logically equivalent
  sentence.

  We may write $\homFrom{\mM}$ as
  \[
    \exists x y_1 \dots y_n .   f(t_1, \dots, t_n) = x \land \beta(x, \vec{y})
  \]
  so that the transformed sentence is
  \[
    \exists y_1 \dots y_n . \beta[x :=  f(t_1, \dots, t_n)]( \vec{y})
  \]
  Suppose $\mP$ satisfies the first sentence with environment
  $ \eta = x \mapsto a, \vec{y} \mapsto \vec{b}$.
  Then $\mP$ satisfies the second sentence with
  $\eta' = \vec{y} \mapsto \vec{b}$, since
  $\mP \models   f(t_1, \dots, t_n) = x$ under $\eta$.

  Suppose $\mP$ satisfies the second sentence with environment
  $ \delta = \vec{y} \mapsto \vec{b}$. %
  Then $\mP$ satisfies the first sentence with %
  $\delta' = x \mapsto f(\delta t_1, \dots \delta t_n), \vec{y}
  \mapsto \vec{b}$ (this is a suitable environment because $x$ does
  not occur in $ f(t_1, \dots, t_n)$.)
\end{proof}

\begin{example}
Start with (suppressing the $\land$ between equations)
\begin{align*}
  \exists x_0x_1 x_2 \; .\; f x_0 &\to x_2 \\
  f x_1 &\to x_0 \\
  f x_2 &\to x_1 \\
  c &\to x_2
\end{align*}
If we work on the equations in the order given we get
\begin{align*}
  \exists x_0x_1 \; .\;   f x_1 &\to x_0 \\
  f f x_0 &\to x_1 \\
  c &\to f x_0
\end{align*}
\begin{align*}
  \exists x_1 \; .\;
  f f f x _1 &\to x_1 \\
  c &\to f f x_1
\end{align*}
and the end result is
$
  \exists x_1 \; .\;
  (f f f x _1 =  x_1 ) \land (c = f f x_1).
$
\end{example}

The order in which we do these rules matters.
Here is an improved algorithm for
applying the rules.  %

Construct a graph in which the nodes are the variables occuring in the
set of equations, and in which, if $fx_1\dots x_n \to x$ is a rule, then
there is an edge from each $x_i$ to $x$.  Then start at the (definitions
for the) sources of this graph, and proceed along the graph, that is
optimal (conjecture).

\begin{example}
With the same starting point as above:
\begin{align*}
  \exists x_0x_1 x_2 \; .\; f x_0 &\to x_2 \\
  f x_1 &\to x_0 \\
  f x_2 &\to x_1 \\
  c &\to x_2
\end{align*}
 Making the graph as defined above we have
$a_2$ is a source, then $a_1$ then $a_0$.
If we do things in that order we get
\begin{align*}
  \exists x_0x_1 \; .\;   f x_0 &\to c \\
  f x_1 &\to x_0 \\
  f c &\to x_1
\end{align*}
\begin{align*}
  \exists x_0 \; .\;   f x_0 &\to c \\
  f f c &\to x_0
\end{align*}
\begin{align*}
 f f f c&\to c
\end{align*}
This is ideal.
\end{example}

There is an interesting connection here with \emph{initial} models,
that is, those models that are free, in the categorical sense, over
the empty set of generators, among all the models of a theory.  %
For such a model, (i) everything named by a closed term, (ii) no
equations between elements that are not forced by $T$, (iii) no
predicate facts true that are not forced by $T$.  %
If \mM is initial then there is a hom from it to \emph{any} model, so the
compressed representation sentence could just be ``true''.

In general the measure of how much a model fails to be initial is
given (i) elements not named by terms, (ii) unforced equations between
elements (iii) unforced atomic sentences.  Capturing those facts about
a model characterizes the homs possible out of it; this is what the
compressed representation does.

\subsubsection{Computing \amin Models}

Our work on homTo and homFrom leads to an algorithm for $a$-minimality.

\begin{algo}[a-Minimize] \label{alg:a-min} \hfill
\begin{itemize}
\item[] \textbf {input:} theory $T$ and model $\mM \models T$

\item[] \textbf{output:}
  model $\mP \models T$ such that
  $\mP$ is a-minimal for $T$ and $\mN \homleq \mM$

\item[] \textbf{initialize:} set \mP to be \mM

\item[] \textbf{while}
  $T' \eqdef T \cup \set{\homTo{\mP}} \cup \set{\avoid{\mP}}$
  is satisfiable
    \begin{itemize}
    \item[] set \mP to be a model of $T'$
    \end{itemize}
  \item[] \textbf{return} \mP

\item[]
\end{itemize}
 \end{algo}

\subsubsection{Computing  a Set-of-Support}

This is another application of the $\avoid{\mA}$ technique.
Given theory $T$ and model \mA, if we construct the theory
$T' \eqdef T \cup \{\avoid{\mA} \}$ then calls to the SMT solver
on theory $T'$ are guaranteed to return models of $T$ outside the
hom-cone of \mA if any exist.   So a set-of-support for $T$ can be
generated by iterating this process.

Completeness of this strategy does not require that the models \mA we work with are minimal.
But if we do work with minimal models there will be fewer iterations.

When $\Sigma$ is a signature, a \emph{profile} for $\Sigma$ is a map
associating a positive integer with each uninterpreted sort of
$\Sigma$.
These numbers will be treated as upper bounds on the sizes of the sets
interpreting sorts in the models we construct.

\begin{algo}[SetOfSupport]\label{alg:us} \hfill
  \begin{itemize}
  \item[] \textbf {input:}  theory $T$ and profile \profile
  \item[] \textbf{output:} a stream
    $\mM_1, \mM_2, \dots$ %
    of minimal models  of $T$ such that %
    for any \profile-model $\mP \models T$, there is some $i$ such that %
    $\mM_i \homleq \mP$.
  \item [] \textbf{initialize:}
    set theory $T^*$ to be $T \cup \{ \prfbounds \}$

  \item[] \textbf{while}
    $T^*$ is satisfiable
    \begin{itemize}
    \item[] let \mM be  minimal model of $T^*$
    \item[] \textbf{output} \mM
    \item[] set $T^*$ to be $T^* \cup \avoid{\mM}$
    \end{itemize}
  \end{itemize}
\end{algo}


\subsection{Working with an SMT Solver}
\label{sec:working-with-solver}

This section treats some of the practicalities of using  a SMT
solver as a tool to build models.

\subsubsection{Getting a Model from the Solver}
\label{sec:getting-model-from}

Once the solver has determined that a theory $T$ is satisfiable, and
computed---internally---a model for $T$, the application must extract
the model from the solver.  But the API for doing this---in the
solvers we are familiar with---is quite restricted.  In any event,
SMT-Lib compliant solvers are not \emph{required} to make this process
particularly convenient.  %
Quoting from the SMT-Lib Standard (v.2.6) \cite{smtlib}
\begin{quote}\em
  The internal representation of the model A is not exposed by the
  solver. Similarly to an abstract data type, the model can be
  inspected only through the three commands below. As a consequence,
  it can even be partial internally and extended as needed in response
  to successive invocations of some of these commands.
\end{quote}
\newpage
The three commands alluded to are
\begin{itemize}
\item \textbf{get-value}, taking a list of closed quantifier-free terms and
  returning a corresponding sequence of terms designating
  \emph{values.}   The notion of ``value'' is theory-specific; for
  example, for the theory of arithmetic the values are the numerals.
\item \textbf{get-assignment}, a certain restricted version of get-value.
\item \textbf{get-model}, returning a list of definitions specifying the
  meanings of the user-defined function symbols.

  In this case the definitions are given in terms of the solver's
  internal representation of model-values.
\end{itemize}

This is inconvenient for us,  for several reasons.  First, for
an uninterpreted sort there is no theory-defined notion of value.
Second, since the solver might create only a partial model internally, Razor may not
have all the information it requires (for example for minimization).

To address this, we first ensure that the language we use to
communicate with the solver has enough ground terms at each sort to
name all elements of a model.  We expand on this point in
Section~\ref{sec:fresh-constants}.  Once this is done, we can query
the solver for the values of \emph{all} functions and predicates, see
Section~\ref{sec:querying}.  Finally, in
Section~\ref{sec:convergent-rep} we present a convenient data
structure for maintaining models.

\subsubsection{The Fresh-Constants Approach}
\label{sec:fresh-constants}

Suppose we have asked the solver to generate a model for theory $T$,
over a signature $\Sigma$.
We first build an enriched theory $T^{+}$ as follows.
\begin{enumerate}
\item  Determine a bound, at each uninterpreted sort $S$, on the number of
elements in the model(s) at the sort.

\item If the bound at sort $S$ is $n$,  add fresh constants %
$\{c-1, \dots, c_n\}$ to the signature, resulting in an expanded
signature $\Sigma^{+}$
\item Add to the theory a set of sentences, one for each uninterpreted sort $S$
  expressing the constraint that every element of sort $S$ is equal to
  one of the $c_i$.  Note that in a given model it may be the case
  that distinct constants name the same model element.
\end{enumerate}

Then the (bounded) models of the original theory are precisely the
reducts to $\Sigma$ of the models of $T^{+}$.

\subsubsection{Querying the Model}
\label{sec:querying}
Since everything is now named by a term (indeed, a constant) we can
work with the solver according to the standard, as follows.
Suppose the solver has determined $T^{+}$ to be satisfiable.
Here's what we do to scrape a model \mM out of the solver.

\begin{enumerate}
\item We know that every element of \mM is named by one of our
  canonical constants $c_i$.
\item First query the solver for (enough) answers to $\mM \models c_i =^? c_j$
  to get a set of representatives for the domains of \mM.
\item For each predicate $R$ and appropriate argument vector $\vec{c}$
  of representatives, query $\mM \models R(\vec{c})$.  %
  The solver will reply ``true'' or ``false'' and so this collection
  of queries defines the meaning of $R$ in the model.

\item For each function $f$ and  appropriate argument vector $\vec{c}$
  and possible answer $c$ of representatives, query $\mM \models
  f(\vec{c}) = c$. %
  The solver will reply ``true'' or ``false'' and so this collection
  of queries defines the meaning of $f$ in the model.

\end{enumerate}

We can then build a ``basic'' model representation
\begin{align*}
  \text{equations} & \quad c_i = c_j \qquad \text{and} \\
  \text{equations} & \quad f \vec{c} = c \qquad \text{and} \\
  \text{facts} & \quad R \vec{c}
\end{align*}
where the $c_i$ range over the Razor constants.

An improved representation is given in the next section.

\subsubsection{Making the model representation convergent}
\label{sec:convergent-rep}
Suppose we put a total order $\succ$ on $K$, the Razor-generated
constants, and declare that for every $f, c$ and $c'$,
$f( \dots, c, \dots) \succ c'$.  We can then make a convergent ground
(terminating and confluent) rewrite system out of a model representation by
\begin{itemize}
\item
 turning each
$C$-equation and $D$-equation $s=t$ into a rewrite rule
$s \to t$ if $s \succ t$.
\item reducing each rule using the others, and
\item iterating this (since the process may create new rules).
\end{itemize}
This process is guaranteed to terminate, in a system of rules:
\begin{align*}
  \text{equations} & \quad f \vec{c} \to c \qquad \text{// $C$-rules} \\
  \text{equations} & \quad c_i \to c_j \qquad \text{// $D$-rules}
\end{align*}
which is ground convergent.

The $D$-equations make an equivalence relation
on the Razor constants.
The constants occurring on the right-hand sides make
a set of canonical representatives of \mM elements.

Since the system is self-reduced, ie, %
each right-hand side of a rule is irreducible by all the rules and
each left-hand side of a rule is irreducible by all the other rules,
the C-rules mention only the canonical representatives.

Also, these are precisely the set of constants that occur in any of the $C$-equations.


\subsection{Alternative Approaches}

In this section we describe alternative approaches to the problems of
model extraction and $a$-minimization.

There are actually two distinct problems addressed here:
\begin{itemize}
\item getting a model out of
  the solver, and
\item $a$-minimization
\end{itemize}

They are described together here because the solutions have been
implemented together in the same version of Razor, and it is easier to
present them together.

\subsubsection{The Enumerated Types Approach to Model Extraction}
\label{sec:enumerated-types}

We describe an an alternative approach to building models, which was
in fact the first method we implemented.  The idea arose as a solution to
problems described in Section~\ref{sec:getting-model-from} for obtaining
a complete representation of a model once the solver has determined
that a theory is satisfiable.

One solution to this problem was described in
Section~\ref{sec:fresh-constants}.  Another is described here, as
Algorithm ET-First.  Roughly speaking, we translate uninterpreted
sorts into enumerated sorts, thereby ensuring that every element of
the sort is named by a term denoting a \emph{value}.   We can then use
the get-value solver function described above to query the solver.

\subsubsection{Another Exhaustive Search Approach to $a$-Minimization}
\label{sec:another-approach-a}

We also show here an approach to $a$-minimization different from
Algorithm~\ref{alg:a-min}.  This algorithm, Algorithm (ET Last),
starts with a given model \mM, and computes a sequence of $i$-minimal
models, each one constructed to be outside the $a$-cone of the
previous.  If and when a final such model \mK is reached, it is
guaranteed to be $a$-minimal.  Since: to say that \mK is a final model
in the sequence is to say that there are no models of $T$ outside the
cone of \mK; and if \mP were strictly below \mK in the $\homleq$
ordering then \mP would violate that property.
If $T$ is a theory with a finite set of support,
this technique is guaranteed to compute an $a$-minimal model.

The use of $i$-minimal models in this method is not essential for its
correctness, but we do this for efficiency: $i$-minimization enlarges
a model's avoid-cone (even though the cone is defined in terms of
$a$-homomorphisms) and so decreases the numer of iterations required
before arriving at a final model.

\begin{algo}[Enumerated Types]\label{alg:et} \hfill
  \begin{itemize}
  \item[] \textbf{input:} theory $T$
  \item[] \textbf{output:} a stream of minimal models comprising a set of
    support for $T$
  \end{itemize}
  \pinch
  \qquad $\cn{ET}(T){}\equiv{}$
  \begin{listalg}
  \item $\mM\gets\cn{ET\_first}(T)$.
  \item If $T$ is unsatisfiable, return.
  \item $\mN\gets\cn{ET\_last}(T,\seq{\mM})$.
  \item Output $\mN$.
  \item $T'\gets T\cup\cn{avoid\_cone}(\mN)$.
  \item $\cn{ET}(T')$.
  \end{listalg}
\end{algo}

\begin{algo}[ET First]\label{alg:et-first} \hfill
  \begin{itemize}
  \item[] \textbf{input:} theory $T$
  \item[] \textbf{output:} a model of $T$ or \cnc{UNSAT}
  \end{itemize}
 \qquad $\cn{ET\_first}(T){}\equiv{}$
  \begin{listalg}
  \item If $T$ is unsatisfiable, return \cnc{UNSAT}.
  \item Extract domain~$D$ from prover.
  \item Replace the uninterpreted sorts in~$T$ with enumerated types
    to form theory~$T'$.  The domain~$D$ is used to determine the
    number of scalar constants within each enumerated type.
  \item Ensure $T'$ is satisfiable.
  \item $\mM\gets\cn{get\_model}(T')$.
  \item Return an $i$-minimization of \mM (cf. Algorithm~\ref{alg:imin}).
  \end{listalg}
\end{algo}

\begin{algo}[ET Last]\label{alg:et-last} \hfill
  \begin{itemize}
  \item[] \textbf{input:} theory $T$
  \item[] \textbf{output:} a set of minimal models comprising a set of
    support for $T$
  \end{itemize}
  \pinch
  \qquad $\cn{ET\_last}($T$){}\equiv{}$
  \begin{listalg}
  \item $T'\gets T\cup\{\cn{avoid\_cone}(\mM)\mid\mM\in M\}$.
  \item $\mN\gets\cn{ET\_first}(T')$.
  \item If $T'$ is unsatisfiable, return $\cn{head}(M)$.
  \item $N\gets{}$ filter out of $M$ models with a homomorphism from~$\mN$.
  \item Return $\cn{ET\_last}(T,\seq{\mN}\append N)$.
  \end{listalg}
\end{algo}

The \cn{avoid\_cone} function returns the avoid cone associated with a
simplified version of the model.  Simplification is crucial for
performance reasons.

\[\avoid{\mM} \eqdef \neg \homFrom{\mM}\]






%% file: results.tex
\subsection{Results}\label{sec:results}

Early implementations of {\lpa} had errors.  At each step in the
{\lpa} algorithm, the model finder must produce a model that describes
a skeleton.  Early implementations were node-oriented, and it took a
while to identify the correct skeleton axioms.  A missing axiom meant
that the skeleton extracted from model finder output was not accepted
by {\cpsa}.  The axioms are now finely tuned to assert only what is
needed to obtain valid skeletons, however, notice that there is a
fairly large number of universally quantified formulas that make up a
protocol theory.

The performance of early implementations of {\lpa} was miserable.  It
was immediately clear that efficient model finding is essential,
motivating the extensive work described in this chapter.  The first
viable algorithm developed (Algorithm~\ref{alg:et}) is called ET for
Enumerated Types, and the second one (Algorithm~\ref{alg:us}) is
called US for Uninterpreted Sorts.

Early work with the ET algorithm showed the performance advantage of
using a strand-oriented goal language, but throughout, we retained the
capability to use a node-oriented language.

A huge performance boost came to the ET algorithm by compressing the
representation of a model a described in Algorithm~\ref{alg:hom
  from}.  The US algorithm always compressed representations.

\begin{table}
  \begin{center}
  \begin{tabular}{|l|rc|rc|}\hline
    \multicolumn{1}{|c}{Test}&\multicolumn{2}{|c|}{ET}
    &\multicolumn{2}{|c|}{US}\\ \cline{2-5}
    \multicolumn{1}{|c|}{Name}&Time&Sat Cks&Time&Sat Cks\\ \hline
    Needham-Schroeder (NO)&1.67&14 14&1.67&8 7\\
    Needham-Schroeder (SO)&1.43&34 12&1.26&9 8\\
    Reflect (NO)&2.82&11 32&1.89&6 9\\
    Reflect (SO)&1.45&10 30&0.56&5 7\\
    DoorSEP (SO)&4.32&26 60 37&4.92&9 10 24\\ \hline
  \end{tabular}
  \end{center}
  \caption{{\lpa} Results}\label{tbl:lpa results}
\end{table}

Intuition suggests that the US algorithm will outperform the ET
algorithm, and tests bear out that intuition.  For each test of
{\lpa}, we collected the \textsc{cpu} runtime in seconds, and for each
invocation of the model checker, we recorded the number of
satisfaction checks requested of Z3.

Table~\ref{tbl:lpa results} shows the results of running {\lpa} on
three protocols: Needham-Schroeder, Reflect, and DoorSEP\@.  Reflect
is the simplest protocol we could come up with that has two shapes.
See Figure~\ref{fig:reflect}.  NO indicates the use of the
node-oriented language and SO is for the strand-oriented language.
Notice that US always makes far fewer satisfaction checks, and tends
to run faster than ET\@.  DoorSEP with the node-oriented language does
not terminate in a reasonable amount of time.  Analyzing DoorSEP
requires three invocations of the model checker.  It is the third
invocation of the model checker that fails to terminate, however the
performance of the second invocation of the model checkers shows the
superiority of the US algorithm. It took 1.65 seconds as opposed to
54.21 seconds used by ET\@.  For the equivalent strand-oriented
problem, it took the US algorithm 1.11 seconds as opposed to 3.41
seconds used by ET\@.  Thus ET was 16 times faster solving the
strand-oriented problem.

Attempts to use {\lpa} on larger protocols proved futile.  Larger
protocols have a larger set of skeleton axioms.  When Z3 looks for a
model, it must instantiate each universally quantified variable with
each element in its domain.  The result is an exponential growth in
the resources used by Z3 to find a model.  Our results show that
na\"ive model finding modulo strand space theory is only viable for
very small protocols.

\begin{figure}
  \[\begin{array}[c]{c@{\qquad\qquad}c}
  \mbox{init}&\mbox{resp}\\[2ex]
  \xymatrix@C=6em@R=2.3ex{
    \bullet\ar@{=>}[d]\ar[r]^{\enc{B}{\iv{A}}}&\\
    \bullet&\ar[l]_{\enc{A}{\iv{B}}}}
  &
  \xymatrix@C=6em@R=2.3ex{
    \ar[r]^{\enc{B}{\iv{A}}}&\bullet\ar@{=>}[d]\\
    &\bullet\ar[l]_{\enc{A}{\iv{B}}}}
  \end{array}\]

  \begin{center}
    Point of view: responder receives $\enc{B}{\iv{A}}$\\
    $\iv{A}$ and $\iv{B}$ are uncompromised
  \end{center}
  \[\fbox{\xymatrix@C=4em@R=2.3ex{
    \txt{\strut resp}&\txt{\strut init}\\
    \bullet&\bullet\ar[l]_{\enc{B}{\iv{A}}}}}\]

  \[\fbox{\xymatrix@C=4em@R=2.3ex{
      \txt{\strut resp}&\txt{\strut resp}&\txt{\strut init}\\
    &\bullet\ar@{=>}[d]&\bullet\ar[l]_{\enc{A}{\iv{B}}}\\
    \bullet&\bullet\ar[l]_{\enc{B}{\iv{A}}}}}\]
  \caption{Reflect Protocol and Shapes}\label{fig:reflect}
\end{figure}

%% file: concl.tex
\section{Conclusion}\label{sec:concl}

In this paper, we have studied the mechanisms needed for finding
minimal models efficiently in the preorder of all homomorphisms, or in
the preorder of embeddings.  We have described an implementation of
these methods in Razor that orchestrates Z3 to compute the models.
Moreover, we have shown how to share labor between Razor and {\cpsa}
so that the latter can apply its authentication test solving methods,
while Razor is handling the remainder of the axiomatic theory of the
protocol together with some non-protocol axioms.

The project explored several algorithms for finding minimal models.  A
significant improvement was described in
Section~\ref{hom-from-and-to}, which explains how to optimize the
construction of the sentence $\avoid{\mA}$ that characterizes the cone
of models to avoid a given model.  Tests show that the
Uninterpreted Sorts algorithm (Algorithm~\ref{alg:us}) outperforms
all others.  This algorithm is implemented in a program that can be
used outside the {\lpa} framework.

The project identified an axiomatic theory for each protocol that is
finely tuned so as to allow Z3 and {\cpsa} to communicate.  The theory
makes it so that well-formed skeletons can be extracted from Razor
models and given to {\cpsa}.  With these theories, we successfully
analyzed the DoorSEP protocol which includes a trust axiom.
Unfortunately, as the size of a protocol grows, so does the size of
its theory, and especially its number of universally quantified
variables.  Z3 becomes very slow when given a large theory.  Thus we
found that the {\lpa} architecture, as currently implemented, cannot
scale to handle nearly all problems of interest.

In future work, we would like to reorganize the software architecture
as well as the selection of logical theories to deliver to the
components.  Z3 is reduced to a molasses-like consistency when given
reasonably large domain sizes and a theory with as many universal
quantifiers as appear in our protocol theories.  This motivates an
architecture in which only subtheories are delivered to Z3, preferably
governing smaller parts of the domain.

Thus, one would like to do more reasoning locally with a successor to
{\cpsa}.  Part of this reasoning can take the same form as generating
the current cohorts, i.e.~applying authentication test-like reasoning
to handle authentic and secure channels, mutable global state, and to
apply security goals ascertained in previous runs.  Moreover, explicit
logical axioms may also be handled in the same way, when they are in
the form of geometric sequents:
\[
  \Phi \limp \bigvee_{i\in I} \exists \overline{y_i} \qdot \Psi_i ,
\]
where $\Phi$ and the $\Psi_i$ are conjunctions of atomic formulas, and
particularly in the favorable case that the index set $I$ is either a
singleton or else the empty set $I=\emptyset$.  In that favorable
case, the inference does not require a case split, but only adding
information.  In these situations, a {\cpsa}-like program can
certainly saturate its skeleton-like partial models.  It can then call
out to Razor to obtain minimal models of portions of the theory that
involve modest domains and limited numbers of nested universal
quantifiers.

